\documentclass[11pt,letterpaper]{article}
\usepackage[utf8]{inputenc}
\usepackage{amsmath}
\usepackage{amsfonts}
\usepackage{amssymb}
\usepackage{amsthm}
\usepackage{algpseudocode}
\usepackage{algorithm}
\usepackage{thmtools}
\usepackage{xcolor}
\usepackage{afterpage}
\usepackage{enumitem}
\usepackage{authblk}

\usepackage{array}
\usepackage{verbatim}
\usepackage{float}
\usepackage{multirow}

\definecolor{ForestGreen}{rgb}{0.1333,0.5451,0.1333}
\definecolor{DarkRed}{rgb}{0.8,0,0}
\definecolor{Red}{rgb}{1,0,0}
\usepackage[linktocpage=true,
pagebackref=true,colorlinks,
linkcolor=DarkRed,citecolor=ForestGreen,
bookmarks,bookmarksopen,bookmarksnumbered]
{hyperref}

\usepackage{cleveref}
\usepackage{geometry}
\newcommand{\R}{\mathbb{R}}
\newcommand{\I}{\mathbb{I}}
\renewcommand{\tilde}{\widetilde}
\renewcommand{\hat}{\widehat}

\newcommand{\diag}{\operatorname{diag}}
\newcommand{\new}{ {\operatorname{new}} }
\newcommand{\polylog}{\operatorname{polylog}}
\newcommand{\tdelta}{\tilde{\delta}}

\algnewcommand{\LineComment}[1]{\State \(\triangleright\) #1}
\algdef{SE}[VARIABLES]{Variables}{EndVariables}
   {\algorithmicvariables}
   {\algorithmicend\ \algorithmicvariables}
\algnewcommand{\algorithmicvariables}{\textbf{global variables}}

\declaretheorem[numberwithin=section]{theorem}
\declaretheorem[numberlike=theorem]{lemma}
\declaretheorem[numberlike=theorem]{proposition}
\declaretheorem[numberlike=theorem]{corollary}
\declaretheorem[numberlike=theorem]{Definition}

\declaretheorem[numberlike=theorem,style=remark]{Remark}
\declaretheorem[numberlike=theorem, refname={task,tasks}]{task}

\title{
A Deterministic Linear Program Solver \\
in Current Matrix Multiplication Time}
\author[1]{Jan van den Brand}
\affil[1]{KTH Royal Institute of Technology, Sweden}
\date{}

\begin{document}

\begin{titlepage}
	\maketitle
	\pagenumbering{roman}
	\ifdefined\ShowComment
	\begin{center}
		{\centering\huge\textcolor{red}{DEBUG VERSION}}
	\end{center}
	\fi
	
Interior point algorithms for solving linear programs 
have been studied extensively for a long time [e.g. 
Karmarkar 1984; 
Lee, Sidford FOCS'14; 
Cohen, Lee, Song STOC'19].
For linear programs of the form 
$\min_{Ax=b, x \ge 0} c^\top x$ 
with $n$ variables and $d$ constraints, 
the generic case $d = \Omega(n)$ has recently been settled 
by Cohen, Lee and Song [STOC'19].
Their algorithm can solve linear programs in 
$\tilde O(n^\omega \log(n/\delta))$ expected time%
\footnote{
Here $\tilde{O}$ hides $\polylog(n)$ factors 
and $O(n^\omega)$ is the time required to multiply two $n \times n$ matrices.
The stated $\tilde O(n^\omega \log(n/\delta))$ bound 
holds for the current bound on $\omega$ with $\omega \approx 2.38$ [V.Williams, STOC'12; Le Gall, ISSAC'14].
The upper bound for the solver will become larger than 
$\tilde O(n^\omega \log(n/\delta))$, if $\omega < 2+1/6$.%
}, where $\delta$ is the relative accuracy.
This is essentially optimal as all known linear system solvers 
require up to $O(n^{\omega})$ time for solving $Ax = b$.
However, for the case of \emph{deterministic} solvers, 
the best upper bound is Vaidya's 30 years old
$O(n^{2.5} \log(n/\delta))$ bound [FOCS'89].
In this paper we show that one can also settle the deterministic setting 
by derandomizing Cohen et al.'s $\tilde{O}(n^\omega \log(n/\delta))$ time algorithm. 
This allows for a strict $\tilde{O}(n^\omega \log(n/\delta))$ time bound, 
instead of an expected one, and a simplified analysis, 
reducing the length of their proof of their 
central path method by roughly half.
Derandomizing this algorithm was also an open question asked in Song's PhD Thesis.

The main tool to achieve our result is a new data-structure 
that can maintain the solution to a linear system in subquadratic time.
More accurately we are able to maintain 
$\sqrt{U}A^\top(AUA^\top)^{-1}A\sqrt{U}\:v$ 
in subquadratic time under $\ell_2$ multiplicative changes 
to the diagonal matrix $U$ and the vector $v$. 
This type of change is common for interior point algorithms.
Previous algorithms [e.g. Vaidya STOC'89; Lee, Sidford FOCS'15; Cohen, Lee, Song STOC'19] 
required $\Omega(n^2)$ time for this task.
In [Cohen, Lee, Song STOC'19] they managed to maintain the matrix 
$\sqrt{U}A^\top(AUA^\top)^{-1}A\sqrt{U}$ in subquadratic time, 
but multiplying it with a dense vector to solve the linear system
still required $\Omega(n^2)$ time.
To improve the complexity of their linear program solver, 
they restricted the solver to only multiply sparse vectors 
via a random sampling argument.
In comparison, our data-structure maintains the entire product 
$\sqrt{U}A^\top(AUA^\top)^{-1}A\sqrt{U}\:v$ 
additionally to just the matrix.
Interestingly, this can be viewed as a simple modification 
of Cohen et al.'s data-structure,
but it significantly simplifies their analysis 
of their central path method
and makes their whole algorithm deterministic.

	\newpage 
	
	\setcounter{tocdepth}{2}
	\tableofcontents
	
\end{titlepage}
\pagenumbering{arabic}

%\end{document}
%\endinput

\section{Introduction}
\label{sec:introduction}

Fast algorithms for solving linear programs 
have a long history in computer science.
Solving linear programs was first proven 
to be in $P$ in 1979 by Khachiyan \cite{Khachiyan79};
and later Karmarkar \cite{Karmarkar84} 
found the first polynomial time algorithm 
that was feasible in practice. 
This initiated the long line of work 
of solving linear programs using interior point algorithms, 
motivated by the fact 
that many problems can be stated as linear programs 
and solved using efficient solvers. %
\cite{%
Renegar88,% reducing iterations
Vaidya87,%
Vaidya89a,% reducing iterations
Vaidya89b,% solve system fast
Megiddo89,
NesterovN89,% reducing iterations
NesterovN91,% solve system fast
VaidyaA93,% reducing iterations
Anstreicher96,
NesterovT97,% barrier function to reduce root n to root d
Anstreicher99,
LeeS14,% reducing iterations
LeeS15% solve system fast
}

For linear programs of the form 
$\min_{Ax=b, x \ge 0} c^\top x$ 
with $n$ variables, $d$ constraints 
and $nnz(A)$ non-zero entries, 
the current fastest algorithms are
$\tilde{O}(\sqrt{d}(nnz(A) + d^2))$ \cite{LeeS14, LeeS15}%
\footnote{%
Here $\tilde{O}( \cdot )$ hides $\text{polylog}(n)$ and $\text{polylog}(1/\delta)$ terms.
}
and $\tilde{O}(n^\omega)$-time \cite{CohenLS19}%
\footnote{%
The algorithm of \cite{CohenLS19} runs in 
$O((n^\omega+n^{2.5-\alpha/2+o(1)}+n^{2+1/6})\log(n) \log(n/\delta))$ time, 
where $\delta$ is the relative accuracy 
and $\alpha$ is the dual matrix exponent.
The dual exponent $\alpha$ is the largest $a$ 
such that an $n \times n$ matrix can be multiplied 
with an $n \times n^a$ matrix 
in $n^{2+o(1)}$ arithmetic operations.
For current $\omega \approx 2.38$ 
and $\alpha \approx 0.31$
this time complexity is just 
$O(n^{\omega}\log (n) \log(n/\delta))$.
}, 
where the bound $O(n^\omega)$ is the number of arithmetic operations 
required to multiply two $n\times n$ matrices.\footnote{%
The parameter $\omega$ is also called \emph{matrix exponent}.
}
For the generic case $d = \Omega(n)$, 
the latter complexity is essentially optimal
as all known linear system solvers 
require up to $O(n^{\omega})$ time 
for solving $Ax = b$.
As the complexity is essentially optimal, but the algorithm is randomized,
a typical next step (e.g. \cite{%
KawarabayashiT19,%
Chazelle00a,%
PettieR02,%
MurtaghRSV17%
}) is to attempt to derandomize this algorithm.
Derandomizing algorithms has the benefit 
that the required analysis can lead to further understanding of the studied problem.
There have been precedences where derandomizing algorithms 
required developing new techniques,
which then allowed for improvements in other settings.
For example, in order to derandomize 
Karger's edge connectivity algorithm \cite{Karger00}, 
Kawarabayashi and Thorup \cite{KawarabayashiT19} 
had to develop new techniques, 
which then lead to new results 
in the distributed setting \cite{DagaHNS19}.

In the related area of linear program solvers in the \emph{real RAM model}
(i.e. when analyzing the complexity only in terms of the dimension, but not the bit-complexity of the input),
a lot of effort has been put in derandomization and finding fast deterministic algorithms 
(see e.g. \cite{%
ChazelleM96,%
BronnimannCM99,%
Chan16}).
Yet, there is still a wide gap between the best randomized and deterministic complexity bounds.%
\footnote{%
For an overview see \cite{Chan16}.
The fastest deterministic algorithm requires $O(n d^{d(1/2+o(1))})$ time \cite{Chan16}, 
while with randomized techniques an $O(nd^2 + \exp(O(\sqrt{d \log d})))$ time bound is possible (a combination of \cite{Clarkson95,Kalai92,MatousekSW96}).
}
The same observation can be made in our setting, 
when analyzing the complexity with respect to the bit-complexity of the input, 
where the best deterministic bounds are
$\tilde{O}(\sqrt{n} \cdot nnz(A) + nd^{1.38})$ \cite{Karmarkar84}%
\footnote{When using the $\tilde{O}(\sqrt{n})$-iterations short step method.},
$\tilde{O}(\sqrt{n} \cdot nnz(A) + n^{1.34} d^{1.15})$ \cite{Vaidya89a} and
$\tilde{O}(d \cdot nnz(A) + d^{\omega+1})$ \cite{Vaidya89b}.%
\footnote{%
For curious readers we recommend \cite{LeeS15}. 
They give a brief overview of these algorithms 
and offer a helpful graph that shows which algorithm is fastest 
for which range of $n,d,nnz(A)$.}
For $d = \Omega(n)$, all deterministic algorithms 
are stuck at $\Omega(n^{2.5})$ time.
Further, these bounds are at least 30 years old and all new algorithms, 
that have been able to improve upon these bounds, 
crucially use randomized techniques. 
This raises the question: 
\emph{Is there a deterministic algorithm that can close the gap between deterministic and randomized complexity bounds, or at least break the 30 years old $\Omega(n^{2.5})$ barrier?}

We are able to answer this question affirmatively
by derandomizing the algorithm of Cohen et al.~\cite{CohenLS19}. 
Our deterministic algorithm is not just able 
to break the 30 years old barrier, 
it even matches one of the fastest randomized bounds of $\tilde{O}(n^\omega)$.
This closes the complexity gap between randomized and deterministic algorithms for large $d$.
More formally, we prove the following result:

\begin{theorem}\label{thm:LP}
Let $\min_{Ax = b, x\ge 0} c^\top x$ be a linear program 
without redundant constraints.
Let $R$ be a bound on $\|x\|_1$ for all $x \ge 0$ with $Ax = b$.
Then for any $0 < \delta \le 1$ we can compute $x \ge 0$ such that
\ifdefined\SODAversion
\begin{align*}
c^\top x \le \min_{Ax=b,x\ge 0} c^\top x + \delta \|c\|_\infty R 
\hspace{10pt} \text{and} \\
\| Ax - b \|_1 \le \delta \left(
R \sum_{i,j} |A_{i,j}| + \|b\|_1
\right)
\end{align*}
\else
$$
c^\top x \le \min_{Ax=b,x\ge 0} c^\top x + \delta \|c\|_\infty R
\hspace{20pt}
\text{and}
\hspace{20pt}
\| Ax - b \|_1 \le \delta \left(
R \sum_{i,j} |A_{i,j}| + \|b\|_1
\right)
$$
\fi
in time
$O(n^\omega \log^2(n) \log(n/\delta)),$
for the current matrix multiplication time with $\omega \approx 2.38$ \cite{Williams12,Gall14}.
\end{theorem}

\begin{Remark}
The real complexity of \Cref{thm:LP} is 
$$O((n^\omega + n^{2.5-\alpha/2+o(1)} + n^{2+1/6+o(1)}) \log^2(n) \log(n/\delta)),$$
which can be simplified to 
$O(n^\omega \log^2(n) \log(n/\delta))$ 
for current values of 
$\omega \approx 2.38$ \cite{Williams12,Gall14},
$\alpha\approx 0.31$ \cite{GallU18}.
For integral $A,b,c$ the parameter $\delta = 2^{-O(L)}$ is enough 
to round the approximate solution of \Cref{thm:LP} to an exact solution. 
Here $L = \log ( 1 +\det_{\max} + \|c\|_\infty + \|b\|_\infty)$
is the bit-complexity, where $\det_{\max}$ is the largest determinant 
of any square submatrix of $A$. \cite{Renegar88,LeeS13}
\end{Remark}

Derandomizing the $\tilde{O}(n^\omega)$ algorithm of \cite{CohenLS19}
was stated by Song as an open question in \cite{Song19}.
In addition to answering this open question,
our techniques also allow us to simplify 
the analysis of the central path method used in \cite{CohenLS19}, 
reducing the length by roughly half.

\paragraph{Technical Ideas}

Interior point algorithms must typically repeatedly compute 
the projection of a certain vector $v$, 
i.e. they must compute $Pv$ where $P$ is a projection matrix. 
It suffices to use an approximation $\tilde{P}$ of $P$, 
and in each iteration the matrix $P$ changes only a bit,
which allowed previous results to maintain the approximation $\tilde{P}$ quickly
(See for example \cite{
Karmarkar84,%
NesterovN89,%
Vaidya89a,%
LeeS15%
}).
A natural barrier for improving linear program solvers 
is the fact that computing $\tilde{P}v$ requires $\Omega(n^2)$ for dense $v$.
This leads to the $\Omega(n^{2.5})$ barrier for linear program solvers,
because a total of $\Omega(\sqrt{n})$ projections must be computed.
Cohen et al.~were able to break this barrier in \cite{CohenLS19} 
by sparsifying $v$ to some approximate $\tilde{v}$ via random sampling
and computing $\tilde{P}\tilde{v}$ instead of $\tilde{P}v$.
Our new approach for breaking this barrier deterministically 
is to maintain the product $\tilde{P}\tilde{v}$ directly
for some approximation $\tilde{v}$ of $v$.
This is the key difference of our linear program solver
compared to previous results, which only maintained $\tilde{P}$. 

We will now outline the difference between 
our deterministic $\tilde{O}(n^\omega)$ solver for linear programs
and the randomized result of Cohen et al.~\cite{CohenLS19}.
They managed to obtain a fast solver for linear programs
by computing the projection $\tilde{P}\tilde{v}$ in subquadratic time 
using two clever tools:
\begin{enumerate}[label=({\arabic*})]
\item They created a data-structure to maintain $\tilde{P}$ in sub-quadratic time,
amortized over $\sqrt{n}$ iterations. \label{item:datastructure}
\item They created a novel stochastic central path method 
which can sparsify the vector $v$ to some approximate $\tilde v$ 
via random sampling. \label{item:centralpath} 
Thus the projection $\tilde{P}\tilde{v}$ could be computed in sub-quadratic worst-case time.
\end{enumerate}

Derandomizing this algorithm seems like a difficult task 
as it is not clear how to obtain a deterministic sparsification of $v$.
Recently \cite{LeeSZ19} derandomized 
the central path method \ref{item:centralpath}, 
so they could extend their linear program solver
to the problem of Empirical Risk Minimization. 
However, in order to achieve $\tilde{O}(n^\omega)$ total time, 
they had to reduce some dimension 
in the representation of $\tilde{P}$ via random sketching, 
which resulted in randomizing the data-structure \ref{item:datastructure}.

In this paper we show how to completely derandomize 
the algorithm of \cite{CohenLS19} via a data-structure
that can maintain the projection $\tilde{P}\tilde{v}$ directly 
for some \emph{dense} approximate $\tilde v \approx v$, 
instead of just maintaining the matrix $\tilde{P}$ as in \cite{CohenLS19}.
This result can be obtained in two different ways. 
One option is to use a dynamic linear system algorithm 
(e.g. \cite{Sankowski04,BrandNS19}) via a black-box reduction,
or alternatively one can interpret the resulting data-structure
as a surprisingly simple extension 
of the data-structure used in \cite{CohenLS19}.
Indeed the algorithmic description 
of the data-structure \ref{item:datastructure} of \cite{CohenLS19}
grows only by a few lines 
(see \Cref{alg:projection_maintenance} in \Cref{sec:projection_maintenance}).

The high level idea of our new data-structure is that the vector $v$ 
can be written as some function $v_i = f(w_i)$,
where the argument vector $w$ does not change much 
between two iterations of the central path method. 
By approximating $w$ by some $\tilde{w}$ 
we can re-use information of the previous iteration 
when computing the projection of $\tilde{v} = f(\tilde{w})$.
One difficulty, that we must overcome, 
is that $\tilde{v} := f(\tilde{w})$ 
is not an approximation of $v = f(w)$ in the classical sense
(i.e. we can not satisfy $\|v - \tilde{v}\| \le \varepsilon \|v\|$ or even $\tilde{v}_i \approx v_i$), 
even if $\tilde{w}$ is an approximation of $w$,
because for non-monotonous $f$, the vectors $v$ and $\tilde{v}$ could point in opposite directions.
This is a problem for the short step central path method,
because these algorithm can be interpreted as some gradient descent and
here $v$ depends on the gradient of some potential function.
So if $\tilde{v}$ points in the opposite direction, 
then the algorithm will actually increase the potential function instead of decreasing it.

To solve this issue, we must also perform some small modifications 
to the short step central path method, 
so it is able to handle our approach of ``approximating" $v$.
The main modification to the short step central path method is
that we measure our progress with respect to the $\ell_\infty$-norm,
similar to \cite{AnstreicherB95,CohenLS19}.
The proof for this will be based on \cite{CohenLS19},
where their random sampling \ref{item:centralpath} $\tilde{v}$ of $v$
can be interpreted as some approximation of $v$.
This allows us to adapt their proof 
to our non-standard way of approximating 
$v = f(w)$ via $\tilde{v} = f(\tilde{w})$. 
The removal of all randomized components from their proof
also allows us to reduce the length of their central path analysis by roughly half.
This reduction of the analysis, 
together with the simple extension 
of their data-structure \ref{item:datastructure} 
to maintain $\tilde{P}\tilde{v}$,
is an interesting difference to other derandomization results, 
where complicated tools need to be created for replacing the randomized components.

\section{Outline}

In this section we outline how our algorithm works 
and how we adapt existing ideas 
such as the short step central path method 
and the projection maintenance.
Readers only interested in verifying our algorithm can skip this overview, 
but reading it can help provide some intuition 
for how the different parts of our algorithm interact
and what difficulties must be solved in order for our algorithm to work.

We start the outline with a brief summary of the short step central path method,
which motivates why we must maintain a certain projection. 
Readers already familiar with the short step central path method
can skip ahead to the next subsection \ref{sub:projection_maintenance}.

In \Cref{sub:projection_maintenance} we describe the task of the projection maintenance,
and how we are able to perform this task quickly by using a certain notion of approximation
(details in \Cref{sec:projection_maintenance}).
The next \Cref{sec:outline:modified_central_path} of the outline explains the difficulties
that we encounter by using this type of approximation, 
and how we are able to solve these problems (details in \Cref{sec:central_path}).

Before outlining our linear program solver, 
we want to quickly define some important notation:
%\paragraph{Arithmetic Notation}
For two $n$ dimensional vectors $v,w$ 
we write $vw$ for the entry-wise product and $v/w$ for the entry-wise division, 
so $(vw)_i := v_i w_i$ and $(v/w)_i := v_i/w_i$.
For a scalar $s$ the product $sv$ is the typical entry-wise product 
and analogously we define $v - s$ as the entry-wise difference, 
so $(v-s)_i := v_i - s$. 
%
%\paragraph{Inequalities}
For two vectors $v,w,$ we write $v \le w$ 
if $v_i \le w_i$ for all $i=1,...,n$.
%
%\paragraph{Norms}
We write $\|v\|_p$ for the $\ell_p$-norm, 
so $\|v\|_p = (\sum_{i=1}^n |v_i|^p)^{1/p}$ for $0 < p < \infty$
and $\|v\|_\infty = \max_{i} |v_i|$.

\subsection{Short Step Central Path Method}
\label{sec:outline:central_path}

We first give a brief summary of the short step central path method. 
Readers familiar with these types of algorithms 
can skip ahead to the next subsection \ref{sub:projection_maintenance}.
\medskip

Consider the linear program
$
\displaystyle\min_{Ax=b,x \ge 0} c^\top x
$
and its dual program
$
\displaystyle\max_{A^\top y \le c} b^\top y
$.\\
Given a feasible dual solution $y$ (a vector $y$ s.t. $A^\top y \le c$), 
we can define the slack vector $s := c - A^\top y$. 
Based on the complementary slackness condition (see e.g. \cite{PapadimitriouS82})
we know a triple $(x,y,s)$ is optimal, 
if and only if
\begin{align*}
x_i s_i &= 0 \text{ for all } i,\\
Ax &= b,\\
A^\top y + s &= c,\\
x_i, s_i &\ge 0 \text{ for all } i.
\end{align*}
If only the last three conditions are satisfied, then we call the triple $(x,y,s)$ feasible. 
Given such a feasible triple, we define the vector $\mu$ such that $\mu_i := x_i s_i$ 
and the complementary slackness theorem motivates 
why we should try to minimize the entries of $\mu$.

It is known how to transform the LP in such a way, 
that we can easily construct a feasible solution triple 
$(x,y,s)$ with $x_is_i \approx 1$ for all $i=1,...,n$ (e.g. \Cref{lem:feasible_LP} \cite{YeTM94}).
Thus for $t := 1$ we have $\mu_i \approx t$.
The idea is to repeatedly decrease $t$ and to modify the solution
$x \leftarrow x + \delta_x, y \leftarrow y+\delta_y, s \leftarrow s+\delta_s$ 
in such a way, that the entries of $\mu$ stay close to $t$.
The change of $\mu$ is given by 
$\mu^\new_i = (x+\delta_x)_i(s+\delta_s)_i = \mu_i + x_i\delta_{s,i} + s_i \delta_{x,i} + \delta_{x,i} \delta_{s,i}$ 
and if $\delta_x,\delta_s$ are small enough, 
this can be approximated via $\mu^\new_i \approx \mu_i + x_i \delta_{s,i} + s_i \delta_{x,i}$.
Thus to change $\mu$ by (approximately) $\delta_\mu$, 
we can solve the following linear system
\begin{align}
X \delta_s + S \delta_x &= \delta_\mu, \label{eq:interior_point_system}\\
A \delta_x &= 0, \notag \\
A^\top \delta_y + \delta_s &= 0, \notag
\end{align}
where $X=\diag(x)$ and $S=\diag(s)$ are diagonal matrices 
with the entries of $x$ and $s$ on the diagonal respectively. 
The solution to this system is given by the following lemma:
\begin{lemma}[{\cite{CohenLS19}}]
\label{lem:solution_LP_system}
The solution for $\delta_x,\delta_s$ in \eqref{eq:interior_point_system} 
is given by
$$
\delta_x = \frac{X}{\sqrt{XS}} (I-P) \frac{1}{\sqrt{XS}} \delta_\mu
\text{ and }
\delta_s = \frac{S}{\sqrt{XS}} P \frac{1}{\sqrt{XS}}\delta_{\mu}.
$$
where
$$P := \sqrt{\frac{X}{S}}A^\top \left( A\frac{X}{S}A^\top \right)^{-1} A\sqrt{\frac{X}{S}}.$$
\end{lemma}
A typical choice for the decrement of $t$ is to multiply it by $1-O(\frac{1}{\sqrt{n}})$, 
which means it takes about $O(\sqrt{n} / \delta)$ iterations 
until $t$ reaches some desired accuracy parameter $\delta > 0$ \cite{Renegar88,Vaidya87}.

For the short step central path method the distance 
between $\mu$ and $t$ is typically measured by
$\sum_{i=1}^n (\mu_i - t)^2 = \| \mu - t \|_2^2$ 
and one tries to maintain $x,s$ in such a way that $\| \mu - t \|_2^2 \le O(t^2)$.
This can be modelled via the potential function $\Phi(x) = \| x \|_2^2$,
and then one tries to maintain $\mu$ such that $\Phi(\mu/t-1) = O(1)$, 
which is equivalent to $\| \mu - t \|_2^2 \le O(t^2)$.
Thus a good choice for $\delta_\mu$ would be a vector 
with the same direction as $-\nabla\Phi(\mu/t-1)$,
as this allows us to decrease the potential, 
which then means the distance between $\mu$ and $t$ is reduced.

\subsection{Projection Maintenance (Details in \Cref{sec:projection_maintenance})}
\label{sub:projection_maintenance}

In this subsection we outline one of the main results of this paper and sketch its proof.
As described in the previous section, 
we must repeatedly compute $Pv$ for 
$P := \sqrt{\frac{X}{S}}A^\top \left( A\frac{X}{S}A^\top \right)^{-1} A \sqrt{\frac{X}{S}}$ 
and $v := \frac{\delta_\mu}{\sqrt{XS}}$, 
where the matrix $A$ describes the constraints of the linear program, 
$X$ and $S$ are diagonal matrices 
that depend on some current feasible solution 
and $\delta_\mu$ is some vector.

Our main result is to maintain an approximation of $Pv$ deterministically 
in $\tilde O(n^{\omega-0.5} + n^{2.5-\alpha})$ amortized time, 
where $\omega$ is the current matrix multiplication exponent 
and $\alpha$ is the dual exponent.
This new data-structure is a simple extension 
of the data-structure presented in \cite{CohenLS19}, 
which was able to maintain an approximation of $P$ within the same time bound, 
but their data-structure required up to $O(n^2)$ time for computing $Pv$ for dense $v$.

The exact statement of our result involves various details, for example how $P$ and $v$ change over time. 
So we first want to describe the task of maintaining $Pv$ in more detail.

\paragraph{The task}

\ifdefined\SODAversion
The matrix $P$
\else
The matrix $P = \sqrt{\frac{X}{S}}A^\top \left( A\frac{X}{S}A^\top \right)^{-1} A \sqrt{\frac{X}{S}}$
\fi
shares a lot of structure between two iterations. 
Indeed only the diagonal matrices $X$ and $S$ change, 
while the matrix $A$ stays fixed. 
Thus for simplicity we define $U := X/S$, 
in which case $P := \sqrt{U}A^\top \left( AUA^\top \right)^{-1} A \sqrt{U}$ 
and only the diagonal matrix $U$ changes from one iteration to the next one.

For this task we would wish for a data-structure 
that can compute $Pv$ for any vector $v$ in $O(n^{\omega-0.5})$ time, 
which with $O(\sqrt{n})$ iterations would then result 
in an $O(n^\omega)$-time solver for linear programs. 
More accurately, we hope for an algorithm that solves the following task:

\begin{task}%[{Desired data-structure}]
\label{idea:perfect_datastructure}
Let $A \in \R^{d \times n}$ be a rank $d$ matrix with $n \ge d$.
We wish for a deterministic data-structure with the following operations
\begin{itemize}
\item \textsc{Initialize}$(A,u,v)$: 
Given matrix $A$ and two $n$ dimensional vectors $u,v$ 
we preprocess the matrix and return
$$
Pv := \sqrt{U}A^\top (A U A^\top)^{-1}A\sqrt{U} v,
$$
where $U = diag(u)$ is the diagonal matrix with $u$ on the diagonal.
\item \textsc{Update}$(u,v)$: Given two $n$ dimensional vectors $u,v$, we must compute
$$
Pv := \sqrt{U}A^\top (A U A^\top)^{-1}A\sqrt{U} v.
$$
\end{itemize}
\end{task}

It is not clear whether a data-structure exists for this task with $O(n^{\omega-0.5})$ update time, 
but due to the very first short step linear program solver by Karmarkar \cite{Karmarkar84} 
it is known that one can relax the requirements. 
Indeed it is enough to use an approximation of $P$.

\paragraph{Relaxation and result}

Due to \cite{Karmarkar84} it is known, 
that it is enough to use an approximation 
$\tilde P := \sqrt{\tilde U}A^\top (A \tilde U A^\top)^{-1}A\sqrt{\tilde U}$
for $(1-\varepsilon) U \le \tilde{U} \le (1+\varepsilon) U$, 
instead of the exact matrix $P$.
We show later in \Cref{sec:central_path}, 
that it is also enough to approximate the vector $v$ via some $\tilde{v}$.
The type of approximation for $v$ is a bit different: 
We show in \Cref{sec:central_path} that we can write 
$v = \delta_\mu/\sqrt{XS}$ as a function of $\mu/t$, 
so $v = f(\mu/t)$. 
We then ``approximate" $v$ via some 
$\tilde{v} := f(\tilde{\mu}/t)$, where 
$(1-\varepsilon) \mu \le \tilde{\mu} \le (1+\varepsilon) \mu$.
Note that thus $\tilde{v}$ itself is not necessarily an approximation of $v$ in the classical sense 
(i.e. $\|v - \tilde{v}\|_2 \gg \varepsilon \| v \|_2$)
and the two vectors might even point in opposite directions.

Motivated by these observations we want to maintain $\tilde{P} \tilde{v}$ instead of $Pv$.
This idea allows for a speed-up, because in each iteration 
we only need to change the entries of $\tilde{u}$ and $\tilde{\mu}$
for which the $(1+\varepsilon)$-approximation condition is broken.
Thus if the vectors $u$ and $\mu$ do not change much per iteration, 
then we only need to change few entries of $\tilde{u}$ and $\tilde{\mu}$.
We prove in \Cref{sub:bound_change} that $u$ and $\mu$ satisfy the following condition:

\ifdefined\SODAversion
\begin{lemma}
\label{lem:overview:change}
% overfull box error when using this as a regular theorem title
\textsc{(Proven in \Cref{sub:bound_change}, \Cref{lem:change_sx})}
\else
\begin{lemma}[{Proven in \Cref{sub:bound_change}, \Cref{lem:change_sx}}]
\label{lem:overview:change}
\fi
Let $(u^k)_{k\ge 1}$ be the sequence of vectors $u$, 
generated by the central path method.
Then $\| (u^{k+1} - u^k)/u^{k} \|_2 \le C$ for all $k$ and some constant $C$.
(A similar statement can be made for $\mu$)
\end{lemma}

Thus, while we are not able to solve \Cref{idea:perfect_datastructure} exactly, 
we do obtain a data-structure that 
(i) maintains the solution approximately, 
and (ii) is fast if $\| (u^{k+1} - u^k)/u^{k} \|_2$ and $\| (\mu^{k+1} - \mu^k)/\mu^{k} \|_2$ are small.

\ifdefined\SODAversion
\begin{theorem}\label{lem:overview:projection_maintenance}
\textsc{(Proven in \Cref{sec:projection_maintenance}, \Cref{lem:projection_maintenance})}
\ifdefined\theimportcounter
\else
\newcounter{importcounter}
\fi
\stepcounter{importcounter}
Let $A \in \R^{d \times n}$ be a full rank matrix with $n \ge d$, 
$v$ be an $n$-dimensional vector 
and $0 < \varepsilon_{mp} < 1/4$ be an accuracy parameter.
Let $f : \R \to \R$ be some function that can be computed in $O(1)$ time, 
and define $f(v)$ to be the vector with $f(v)_i := f(v_i)$.
Given any positive number $a \le \alpha$
there is a deterministic data-structure with the following operations
\begin{itemize}
\item \textsc{Initialize}$(A,u,f,v,\varepsilon_{mp})$: 
The data-structure preprocesses the given two $n$ dimensional vectors $u,v$, 
the $d\times n$ matrix $A$ the function $f$ in $O(n^2d^{\omega-2})$ time.
The given parameter $\varepsilon_{mp} > 0$ 
specifies the accuracy of the approximation.
\item \textsc{Update}$(u,v)$: 
Given two $n$ dimensional vectors $u,v$. 
Then the data-structure returns four vectors
$$
\tilde{u}, \hspace{10pt}
\tilde{v}, \hspace{10pt}
f(\tilde{v}), \hspace{10pt}
\sqrt{\tilde U}A^\top (A\tilde U A^\top)^{-1}A\sqrt{\tilde U} f(\tilde v).
$$
\end{itemize}
Here $\tilde{U}$ is the diagonal matrix $\diag(\tilde u)$ and $\tilde{v}$ a vector such that
$$
(1-\varepsilon_{mp}) \tilde{v}_i
\le
v_i
\le
(1+\varepsilon_{mp}) \tilde{v}_i
$$
$$
(1-\varepsilon_{mp}) \tilde{u}_i
\le
u_{i}
\le
(1+\varepsilon_{mp}) \tilde{u}_i.
$$
If the update sequence $u^{(1)},...,u^{(T)}$ (and likewise $v^{(1)},...,v^{(T)}$) satisfies
\ifdefined\SODAversion
\begin{align}
\label{eq:small_change\theimportcounter}
\sum_{i=1}^n \left(\frac{u^{(k+1)}_i - u^{(k)}_i}{u^{(k)}_i}\right)^2 \le C^2,\\
|\frac{u^{(k+1)}_i - u^{(k)}_i}{u^{(u)}_i}| \le 1/4, \notag
\end{align}
\else
\begin{align}
\sum_{i=1}^n \left(\frac{u^{(k+1)}_i - u^{(k)}_i}{u^{(k)}_i}\right)^2 \le C^2,
\hspace{25pt}
|\frac{u^{(k+1)}_i - u^{(k)}_i}{u^{(u)}_i}| \le 1/4,
\label{eq:small_change\theimportcounter}
\end{align}
\fi
for all $k=1,...,T$ then the total time for the first $T$ updates is
$$
O\left(T \cdot \left(C / \varepsilon_{mp} (n^{\omega-1/2} + n^{2-a/2+o(1)}) \log n + n^{1+a}\right)\right)
$$
\end{theorem}
\else
\begin{theorem}[{Proven in \Cref{sec:projection_maintenance}, \Cref{lem:projection_maintenance}}]\label{lem:overview:projection_maintenance}

\end{theorem}
\fi

There are two equivalent ways to prove \Cref{lem:overview:projection_maintenance}: 
One could use the data-structures of \cite{Sankowski04,BrandNS19} 
which maintain $M^{-1}b$ for some non-singular matrix $M$ and some vector $b$. 
Via a black-box reduction these data-structures would then be able to maintain $\tilde{P}\tilde{v}$ 
and applying the tools of \cite{CohenLS19} for optimizing the amortized complexity would then result in \Cref{lem:overview:projection_maintenance}.

If one tries to write down a pseudo-code description of the resulting data-structure, 
then the code is very similar to the data-structure from \cite{CohenLS19}. 
This is because all these data-structures are based on the Sherman-Morrison-Woodburry identity.
Hence an alternative way to prove \Cref{lem:overview:projection_maintenance} 
is to take the data-structure from \cite{CohenLS19},
which already maintains $\tilde{P}$, 
and extend such that it also maintains $\tilde{P}\tilde{v}$.

In this paper we present the second option, 
where we modify the existing data-structure of \cite{CohenLS19}. 
This is because we want to highlight that our derandomization result 
can be obtained from a simple modification of the existing randomized algorithm. 
Though for the curious reader we also give a sketch of the first variant in \Cref{app:matrix_inverse}.

\paragraph{Proof sketch (Details in \Cref{sec:projection_maintenance})}

We now outline how to obtain \Cref{lem:overview:projection_maintenance} 
by extending the data-structure of \cite{CohenLS19} 
to also maintain $\tilde{P}\tilde{v}$, instead of just $\tilde{P}$.
Their data-structure internally has three matrices $M,L,R$
with the property 
\begin{align}
\tilde P = M + LR^\top \label{eq:implicit}
\end{align}
where $M$ is some $n \times n$ matrix and 
$L,R$ are rectangular matrices with some $m \ll n$ columns.
With each update, the matrices $L,R$ change 
and the number of their columns may increase.
This way the $n^2$ entries of the matrix $\tilde{P}$ 
are not explicitly computed and a sub-quadratic update time can be achieved.

As the number of columns $m$ of $L,R$ grows,
the data-structure will become slower and slower.
Once these matrices have too many columns, 
the data-structure performs a ``reset". 
This means we set
\begin{align}
M \leftarrow M + LR^\top, \label{eq:reset}
\end{align} 
and the matrices $L,R$ are set to be empty (so zero columns).
Thus after the reset we have $\tilde{P} = M + LR^\top = M$, 
so \eqref{eq:implicit} is still satisfied.
Such a reset requires $\Omega(n^2)$ time, 
but it does not happen too often so the cost is small on average.\footnote{
Section 5 of \cite{CohenLS19} is about bounding this amortized cost.
}

One can now easily maintain $\tilde{P}f(\tilde{v})$ as follows:
Assume we already know $Mf(\tilde{v})$, 
then a new solution $\tilde{P}f(\tilde{v})$ is given by
\begin{align}
\tilde{P}f(\tilde{v}) = Mf(\tilde{v}) + LR^\top f(\tilde{v}), \label{eq:update_M}
\end{align}
because of \eqref{eq:implicit}.
Here the term $LR^\top f(\tilde{v})$ can be computed in $O(nm) \ll O(n^2)$ time, 
because $L,R$ have $m \ll n$ columns.
The assumption, that $Mf(\tilde{v})$ is known, can be satisfied easily:
During the initialization of the algorithm we compute this value,
and whenever $M$ changes (i.e. during the reset \eqref{eq:reset}) 
we can compute the new $Mf(\tilde{v})$ in $O(n^2)$ time.
This does not affect the complexity of the data-structure, 
because a reset does already require $\Omega(n^2)$ time to compute the new $M$.

At last, we must handle the case where entries of $\tilde{v}$ are changed.
Let's say $\tilde{v}^\new \leftarrow \tilde{v} + \delta_v$, then
\ifdefined\SODAversion
\begin{align*}
\tilde{P}f(\tilde{v}^\new) 
&= \tilde{P}f(\tilde{v}) + \tilde{P} (f(\tilde{v}^\new) - f(\tilde{v})) \\
&= \tilde{P}f(\tilde{v}) + M (f(\tilde{v}^\new) - f(\tilde{v})) \\
&\hspace{20pt} + LR^\top (f(\tilde{v}^\new) - f(\tilde{v})),
\end{align*}
\else
$$\tilde{P}f(\tilde{v}^\new) = \tilde{P}f(\tilde{v}) + \tilde{P} (f(\tilde{v}^\new) - f(\tilde{v})) = \tilde{P}f(\tilde{v}) + M (f(\tilde{v}^\new) - f(\tilde{v})) + LR^\top (f(\tilde{v}^\new) - f(\tilde{v})),$$
\fi
where the last equality comes from \eqref{eq:implicit}.
The complexity can be bounded as follows:
The term $\tilde{P}f(\tilde{v})$ 
is computed as described in \eqref{eq:update_M}. 
The second term $M (f(\tilde{v}^\new) - f(\tilde{v}))$ can be computed quickly 
because on average $\tilde{v}^\new$ and $\tilde{v}$ differ in only few entries, 
because of the small change to $v$ per iteration 
(as given by \eqref{eq:small_change\theimportcounter} of \Cref{lem:overview:projection_maintenance}).
The last term $LR^\top (f(\tilde{v}^\new) - f(\tilde{v}))$ is again computed quickly
because the matrices $L,R$ have very few columns.

\subsection{Adapting the Central Path Method 
for Approximate Projection Maintenance (Details in \Cref{sec:central_path})}
\label{sec:outline:modified_central_path}

We now outline difficulties that occur, 
if one tries to use the projection maintenance algorithm 
(\Cref{lem:overview:projection_maintenance}, outlined in \Cref{sub:projection_maintenance})
in the classical central path method (outlined in \Cref{sec:outline:central_path}),
and how we are able to solve these issues in \Cref{sec:central_path}.

The central path method can be interpreted as some gradient descent, 
where we try to minimize some potential.
When we use the data-structure of \Cref{lem:overview:projection_maintenance},
then we are essentially performing this gradient descent while using some approximate gradient.
This approximation is of such low quality, 
that the approximate gradient occasionally points in a completely wrong direction,
effectively increasing the potential instead of decreasing it.
By adapting the potential function, we are able to prove 
that the approximate gradient only points in the wrong direction 
when the potential is small. 
Whenever the potential is large, 
the approximate gradient points in the correct direction.
(A formal proof of this will be in 
\Cref{sub:mu_close_to_t}, \Cref{lem:gradient_direction}.)
This adaption to the short step central path method 
allows us to handle these faulty approximate gradients.
Before we can outline why this is true, 
we must first explain why we obtain these faulty approximate gradients in the first place.

\paragraph{Faulty gradients}
The central path method tries to maintain some vector $\mu$ close to a scalar $t$, 
where the relative distance is measured via some potential function $\Phi(\mu/t-1)$.
The central path method tries to minimize this potential function 
by solving some linear system that depends on the gradient $\nabla \Phi(\mu/t-1)$.

In \Cref{sec:central_path:algorithm} we show that 
when solving this linear system via \Cref{lem:overview:projection_maintenance}, 
then we are essentially solving the system for the approximation $\nabla\Phi(\tilde{\mu}/t-1)$, 
where $\tilde{\mu}$ is an approximation of $\mu$ 
with $(1-\varepsilon_{mp}) \tilde{\mu} \le \mu \le (1+\varepsilon_{mp}) \tilde{\mu}$
for some accuracy parameter $\varepsilon_{mp} > 0$.
This is problematic because $\nabla \Phi(\mu/t-1)$ and $\nabla\Phi(\tilde{\mu}/t-1)$ 
could point in opposite directions.
For example for any $i$ with $\mu_i > t$ we might have $\tilde{\mu}_i < t$,
so since $\Phi$ tries to keep $\mu$ close to $t$, 
the approximate gradient can not reliably tell if $\mu_i$ should be increased or decreased.
However, if $\mu_i > (1+\varepsilon_{mp}) t$, then $\tilde{\mu}_i > t$, 
so the approximate gradient will correctly try to decrease $\mu_i$.
On one hand this shows, that we can not use the classical short step central path method,
where one tries to maintain $\mu$ such that $\| \mu/t - 1 \|_2^2 = O(1)$. 
This is because $\| \mu/t - 1 \|_2^2$ could be as large as $\Omega(n \varepsilon_{mp})$.
On the other hand, we are able to prove in \Cref{sub:mu_close_to_t} that $\| \mu/t - 1 \|_\infty = O(1)$,
because once some entry $\mu_i$ is further from $t$ than some $(1\pm \varepsilon_{mp})$-factor, 
then the approximate gradient will correctly try to move $\mu_i$ closer to $t$.
Hence, we adapt the short step central path method 
by guaranteeing $\mu$ close to $t$ in $\ell_\infty$-norm, instead of $\ell_2$-norm.

\paragraph{Adapting the central path method}

Luckily, maintaining $\mu$ close to $t$ in $\ell_\infty$-norm 
was previously done in \cite{CohenLS19},
so we can simply adapt their proof for our algorithm.
The high-level idea is to use 
$\Phi(x) = \sum_{i=1}^n (e^{\lambda x_i} + e^{-\lambda x_i})/2$ 
for some parameter $\lambda = \Theta(\log n)$ as the potential function. 
This potential is useful because $\|x\|_\infty \le \lambda^{-1} \log 2\Phi(x)$ 
(proven in \Cref{lem:bound_infty_via_phi}).
This means bounding $\Phi(\mu/t-1)$ by some polynomial in $n$ 
is enough to prove $\|\mu/t-1\|_\infty = O(1)$,
which will be done in \Cref{sub:mu_close_to_t}.

The majority of the proof that this choice for $\Phi$ works,
is adapted from \cite{CohenLS19}. 
For their \emph{stochastic central path method}, 
Cohen et al.~sparsify the gradient $\nabla \Phi(\mu/t-1)$
via randomly sampling its entries.
This sparsification could be interpreted as some type of approximation
of the gradient, which allows us to adapt their proof 
to our new notion of ``approximating" the gradient 
via $\nabla \Phi(\tilde{\mu}/t-1)$ for $\tilde{\mu} \approx \mu$.
The main difference is that in \cite{CohenLS19},
the exact and approximate gradient always point in the same direction
(i.e. their inner product is positive),
so in \cite{CohenLS19} it was a bit easier to show that the potential $\Phi(\mu/t-1)$ decreases in each iteration.
For comparison, when using our approximation, 
the inner product of exact gradient $\nabla\Phi(\mu/t-1)$ 
and the approximation $\nabla\Phi(\tilde{\mu}/t-1)$ 
may become negative.
So we must spend some extra effort in \Cref{sub:mu_close_to_t} 
to show that the approximate gradient points in the correct direction, 
whenever $\Phi(\mu/t-1)$ is large (this will be proven in \Cref{lem:gradient_direction}).
Intuitively, this is true because when $\Phi(\mu/t-1)$ is large,
then there are many indices $i$ such that $\mu_i$ is further from $t$ than some $(1\pm \varepsilon_{mp})$-factor.
As outlined before, this means the approximate gradient tries to change 
the $i$th coordinate of $\mu$ in the correct direction, 
i.e. $i$th entry of the exact and approximate gradient have the same sign.

We also want to point out, 
that our approach of using a gradient 
w.r.t the approximate $\tilde{\mu} \approx \mu$
means it is enough to maintain 
an approximate $\tilde{x} \approx x$, $\tilde{s} \approx s$, 
so $\tilde{x}\tilde{s} =: \tilde{\mu} \approx \mu$.
The same observation was made independently in \cite{LeeSZ19}, 
where that property was exploited to compute 
the steps $\delta_x$, $\delta_s$ approximately via random sketching.
The analysis of their central path method is based on %self-concordant barrier functions and
modifying the standard newton steps 
to be a variant of gradient descent in some hessian norm space.
In comparison our proof is arguably simpler, 
as we perform a typical gradient descent w.r.t $\Phi(\tilde{\mu}/t)$.

\section{Preliminaries}

For the linear program $\min_{Ax = b, x \ge 0} c^\top x$ 
we assume there are no redundant constraints, 
i.e. the matrix $A$ is of rank $d$ and $n \ge d$.

\paragraph{Arithmetic Notation}
For two $n$ dimensional vectors $v,w$ 
their inner products is written as $v^\top w$ 
or alternatively $\langle v, w \rangle$.
We write $vw$ for the entry-wise product, so $(vw)_i := v_i w_i$.
The same is true for all other arithmetic operations, 
for example $(v/w)_i := v_i/w_i$ and $(\sqrt{v})_i := \sqrt{v_i}$.
For a scalar $s$ the product $sv$ is the typical entry-wise product 
and analogously we define $v - s$ as the entry-wise difference, 
so $(v-s)_i := v_i - s$.

\paragraph{Inequalities}
We write $v \le w$ if $v_i \le w_i$ for all $i=1,...,n$ 
and we use the notation $v \approx_\varepsilon w$ 
to express a $(1\pm \varepsilon)$ approximation, 
defined as $(1-\varepsilon) w \le v \le (1+\varepsilon) w$.
Note that $v \approx_\varepsilon w$ is not symmetric, 
but it implies $w \approx_{2\varepsilon} v$ for $\varepsilon \le 1/2$.

\paragraph{Relative error and multiplicative change}

We will often bound the relative difference of two vectors in $\ell_2$-norm:
$\| (v - w) / w \|_2 = (\sum_{i=1}^n ((v_i - w_i) / w_i)^2)^{0.5}.$
Sometimes we will also write this as $\| v/w - 1 \|_2$. 
If we have $v^\new = v + \delta_v$, 
then the relative difference $\|v^\new/v - 1\|_2$ can also be written as
$\| v^{-1} \delta_v \|_2$.
In this context the relative difference will also be called multiplicative change,
because $v^\new = v \cdot (1 + v^{-1} \delta_v)$.

The multiplicative change of a product of two vectors, whose multiplicative change is bounded in $\ell_2$-norm, 
can also be bounded:

\begin{lemma}\label{lem:product_change}
Let $v,w,\delta_v,\delta_w$ be vectors, 
such that $v^\new = v+\delta_v$, $w^\new = w+\delta_w$
then
\ifdefined\SODAversion
\begin{align*}
\| \frac{v^\new w^\new}{vw} -1 \|_2 \le&~ \|v^{-1}\delta_v\|_2 + \|w^{-1}\delta_w\|_2 \\
& + \|v^{-1}\delta_v\|_2\|w^{-1}\delta_w\|_2
\end{align*}
\else
$$\| \frac{v^\new w^\new}{vw} -1 \|_2 
\le \|v^{-1}\delta_v\|_2 + \|w^{-1}\delta_w\|_2 + \|v^{-1}\delta_v\|_2\|w^{-1}\delta_w\|_2$$
\fi
\end{lemma}

\begin{proof}

\ifdefined\SODAversion
\begin{align*}
&\| \frac{v^\new w^\new}{vw} - 1 \|_2
= \| \frac{v^\new w^\new - vw}{vw} \|_2 \\
=&~ \| \frac{(v+\delta_v)(w+\delta_w)-vw}{vw} \|_2 
= \| \frac{v\delta_w + w\delta_v + \delta_v \delta_w}{vw} \|_2 \\
=&~ \| \frac{\delta_w}{w} + \frac{\delta_v}{v} + \frac{\delta_v}{v}\frac{\delta_w}{w} \|_2 
\le \| \frac{\delta_w}{w}\|_2 + \| \frac{\delta_v}{v} \| + \| \frac{\delta_v}{v}\frac{\delta_w}{w} \|_2
\end{align*}
\else
\begin{align*}
\| \frac{v^\new w^\new}{vw} - 1 \|_2
&= \| \frac{v^\new w^\new - vw}{vw} \|_2
= \| \frac{(v+\delta_v)(w+\delta_w)-vw}{vw} \|_2 
= \| \frac{v\delta_w + w\delta_v + \delta_v \delta_w}{vw} \|_2 \\
&= \| \frac{\delta_w}{w} + \frac{\delta_v}{v} + \frac{\delta_v}{v}\frac{\delta_w}{w} \|_2 
\le \| \frac{\delta_w}{w}\|_2 + \| \frac{\delta_v}{v} \| + \| \frac{\delta_v}{v}\frac{\delta_w}{w} \|_2
\end{align*}
\fi
Here the last term can be bounded via
$
\|  \frac{\delta_v}{v}\frac{\delta_w}{w} \|_2 
\le \|  \frac{\delta_v}{v} \|_\infty \|\frac{\delta_w}{w}  \|_2 
\le \|  \frac{\delta_v}{v} \|_2 \|\frac{\delta_w}{w}  \|_2 
$
\end{proof}

Further, if $v^\new$ has small multiplicative change compared to $v$, then the same is true for $1/v^\new$ and $1/v$.

\begin{lemma}\label{lem:inverse_change}
$$\|\frac{(v+\delta_v)^{-1} - v^{-1}}{v^{-1}} \|_2 \le \frac{\|v^{-1}\delta_v\|_2}{1-\|v^{-1}\delta_v\|_2}$$
\end{lemma}

\begin{proof}
\ifdefined\SODAversion
\begin{align*}
&\| \frac{v^{-1} - (v+\delta_v)^{-1}}{v^{-1}} \|_2
=
\| 1 - \frac{v}{v+\delta_v} \|_2 \\
=&~
\| \frac{\delta_v}{v+\delta_v} \|_2 
=
\| v^{-1}\delta_v\frac{v}{v+\delta_v} \|_2 \\
\le&~
\| v^{-1}\delta_v\frac{1}{1-\|v^{-1}\delta_v\|_\infty} \|_2
\le
\frac{\|v^{-1}\delta_v\|_2}{1-\|v^{-1}\delta_v\|_2}
\end{align*}
\else
\begin{align*}
\| \frac{v^{-1} - (v+\delta_v)^{-1}}{v^{-1}} \|_2
&=
\| 1 - \frac{v}{v+\delta_v} \|_2
=
\| \frac{\delta_v}{v+\delta_v} \|_2 
=
\| v^{-1}\delta_v\frac{v}{v+\delta_v} \|_2 \\
&\le
\| v^{-1}\delta_v\frac{1}{1-\|v^{-1}\delta_v\|_\infty} \|_2
\le
\frac{\|v^{-1}\delta_v\|_2}{1-\|v^{-1}\delta_v\|_2}
\end{align*}
\fi
\end{proof}

\paragraph{Fast Matrix Multiplication}

We write $O(n^\omega)$ for the arithmetic complexity of multiplying two $n \times n$ matrices. 
Computing the inverse has the same complexity. 
The exponent $\omega$ is also called the matrix exponent. We call $\alpha$ the dual matrix exponent, 
which is the largest value such that multiplying a $n \times n$ matrix with an $n \times n^\alpha$ 
requires $O(n^{2+o(1)})$ time. 
The current best bounds are 
$\omega \approx 2.38$ \cite{Williams12,Gall14} 
and $\alpha \approx 0.31$ \cite{GallU18}.

\section{Projection Maintenance}
\label{sec:projection_maintenance}

In this section we prove \Cref{lem:projection_maintenance}, 
which specifies the result obtained by \Cref{alg:projection_maintenance}.
Given a matrix $A$, diagonal matrix $U$, vector $v$ and function $f : \R \to \R$ (with $f(v)_i := f(v_i)$),
the data-structure given by \Cref{alg:projection_maintenance}/\Cref{lem:projection_maintenance} 
maintains the solution $\sqrt{U} A^\top(A U A^\top)^{-1} A \sqrt{U} f(v)$ in an approximate way, 
by $(1\pm\varepsilon_{mp})$-approximating $U$ and $v$.
This is an extension of the algorithm from \cite{CohenLS19}, 
which maintained only the matrix $\sqrt{U} A^\top(A U A^\top)^{-1} A \sqrt{U}$ approximately.
We restate the formal description of the result for convenience:

\ifdefined\SODAversion
\begin{lemma}\label{lem:projection_maintenance}
\textsc{(Previously stated as \Cref{lem:overview:projection_maintenance})}

\end{lemma}
\else
\begin{lemma}[{Previously stated as \Cref{lem:overview:projection_maintenance} 
in \Cref{sub:projection_maintenance}}]
\label{lem:projection_maintenance}

\end{lemma}
\fi

This section is split into three parts:
We first present the algorithm and give a high-level description in \Cref{sec:projection:outline}.
The next subsection (\Cref{sec:projection:correctness}) proves 
that the algorithm returns the correct result,
and at last in \Cref{sec:projection:complexity} we bound the complexity of the algorithm.

\subsection{Outline of \Cref{alg:projection_maintenance}}
\label{sec:projection:outline}

\Cref{alg:projection_maintenance} describes a data-structure, 
so we have variables that persist between calls to its function \textsc{Update}.
What these variables represent might be a bit hard to deduce from just reading the pseudo-code,
so we want to give a brief outline of \Cref{alg:projection_maintenance} here.
This outline is not required for verifying the proofs, 
but it might help for understanding how the algorithm works. 

The internal variables are $n$-dimensional vectors $\tilde{u},\tilde{v},w$
and an $n \times n$ matrix $M$. 
The relationship between them is 
\begin{align}
M = A^\top(A \tilde{U} A^\top)^{-1} A \text{ and } w = M \sqrt{\tilde{U}}f(\tilde{v}), \label{eq:consistent_variables_intro}
\end{align}
where $\tilde{U} = \diag(\tilde{u})$.

These internal variables are useful because of the following reason:
In each call to \textsc{Update}, 
the data-structure receives two new vectors $u^\new, v^\new$
and for $U^\new = \diag(u^\new)$ the task is to return an approximation of
$\sqrt{U^\new} A^\top(A U^\new A^\top)^{-1} A \sqrt{U^\new} f(v^\new)$
by $(1\pm\varepsilon_{mp})$-approximating $U^\new$ and $v^\new$.
Thus if
\begin{align}
u^\new \approx_{\varepsilon_{mp}} \tilde{u},
v^\new \approx_{\varepsilon_{mp}} \tilde{v}, \label{eq:approximation_condition}
\end{align}
then $\sqrt{\tilde{U}}w$ would be the desired approximate result.
If this $(1+\varepsilon_{mp})$-approximation condition \eqref{eq:approximation_condition} is not satisfied, 
then we can define two new valid approximations
\ifdefined\SODAversion
\begin{align*}
\tilde{u}^{\new}_{i} :=
\begin{cases} 
	\tilde{u}_{i} & \text{if } u^\new_i \approx_{\varepsilon_{mp}} \tilde{u}_i\\
	u^{\new}_{i} & \text{othewise}
\end{cases} \\
\tilde{v}^{\new}_{i} :=
\begin{cases} 
	\tilde{v}_{i} & \text{if } v^\new_i \approx_{\varepsilon_{mp}} \tilde{v}_i\\
	v^{\new}_{i} & \text{otherwise}
\end{cases}
\end{align*}
\else
\begin{align*}
\tilde{u}^{\new}_{i} :=
\begin{cases} 
	\tilde{u}_{i} & \text{if } u^\new_i \approx_{\varepsilon_{mp}} \tilde{u}_i\\
	u^{\new}_{i} & \text{othewise}
\end{cases}
\hspace{10pt}
\tilde{v}^{\new}_{i} :=
\begin{cases} 
	\tilde{v}_{i} & \text{if } v^\new_i \approx_{\varepsilon_{mp}} \tilde{v}_i\\
	v^{\new}_{i} & \text{otherwise}
\end{cases}
\end{align*}
\fi
for all $i=1,...,n$. 
If $\tilde{u}^\new$ and $\tilde{u}$ 
(and respectively $\tilde{v}^\new$, $\tilde{v}^\new$) 
differ in at most $k$ many entries,
then it is known (via Sherman-Morrison-Woodbury identity \Cref{lem:woodbury}) 
that one can quickly construct two $n \times k$ matrices $R,L$ 
such that
$$
A^\top(A \tilde{U}^\new A^\top)^{-1} A = M - RL^\top.
$$
This in turn means that we can get the desired approximate result as follows:
\ifdefined\SODAversion
\begin{align*}
&\sqrt{\tilde{U}^\new}A^\top(A \tilde{U}^\new A^\top)^{-1} A \sqrt{\tilde{U}^\new} f(\tilde{v}^\new) \\
=&~
\sqrt{\tilde{U}^\new}(M - RL^\top) \sqrt{\tilde{U}^\new} f(\tilde{v}^\new)
\end{align*}
\begin{align*}
=&~
\sqrt{\tilde{U}^\new} \left(\underbrace{M \sqrt{\tilde{U}} f(\tilde{v})}_{=w}\right. \\
&~ + M \underbrace{\left(\sqrt{\tilde{U}^\new} f(\tilde{v}^\new) - \sqrt{\tilde{U}} f(\tilde{v})\right)}_{\text{at most $2k$ non-zero entries}} \\
&~ \left. - RL^\top \sqrt{\tilde{U}^\new} f(\tilde{v}^\new)\right)
\end{align*}
\else
\begin{align*}
&\sqrt{\tilde{U}^\new}A^\top(A \tilde{U}^\new A^\top)^{-1} A \sqrt{\tilde{U}^\new} f(\tilde{v}^\new)
=
\sqrt{\tilde{U}^\new}(M - RL^\top) \sqrt{\tilde{U}^\new} f(\tilde{v}^\new) \\
=&\:
\sqrt{\tilde{U}^\new} \left(\underbrace{M \sqrt{\tilde{U}} f(\tilde{v})}_{=w}
+ M \underbrace{\left(\sqrt{\tilde{U}^\new} f(\tilde{v}^\new) - \sqrt{\tilde{U}} f(\tilde{v})\right)}_{\text{at most $2k$ non-zero entries}}
- RL^\top \sqrt{\tilde{U}^\new} f(\tilde{v}^\new)\right)
\end{align*}
\fi
Here each term can be computed in at most $O(nk)$ time, 
because the first term is the already known vector $w$, 
the vector of the second term is sparse,
and $R,L$ are $n \times k$ matrices.

Thus, if $k$ is small, then we can maintain the solution quickly.
In \cite{CohenLS19}, Cohen et al.~%
have developed a strategy with low amortized cost, 
that specifies when to recompute $M$ for some new $\tilde{u}$,
such that $k$ stays small. 
In their algorithm they do not maintain the matrix-vector product, 
so their data-structure does not have the internal variables $w$ and $\tilde{v}$.
We extend their strategy to also recompute $w$ for some new $\tilde{v}$,
such that the above outlined procedure has low amortized cost.

\begin{algorithm*}
\caption{Projection Maintenance Data-Structure (difference to \cite{CohenLS19} highlighted in blue)}\label{alg:projection_maintenance} {\footnotesize
\begin{algorithmic}[1]
\State {\bf datastructure} \textsc{MaintainProjection} \Comment{\Cref{lem:projection_maintenance}}
\State {\bf members}
  \State \hspace{4mm} $\tilde{u},\tilde{v}, w \in \R^{n}$
  \State \hspace{4mm} $f : \R \to \R$
  \State \hspace{4mm} $A \in \R^{d \times n}$,
  %\State \hspace{4mm} 
  $M \in \R^{n \times n}$  %\Comment{Approximate projection matrix: $M = A^\top(A\tilde{U}A^\top)^{-1}A$}
  \State \hspace{4mm} $\epsilon_{mp} \in (0,1/4)$ \Comment{Accuracy parameter}
  \State \hspace{4mm} $a \leftarrow \min\{\alpha,2/3\}$ \Comment{Minimum batch size is $n^a$.}
\State {\bf end members}
 %\smallskip
  \Procedure{\textsc{Initialize}}{$A,u,f,v,\epsilon_{mp}$}
  	\State $u \leftarrow u$, $v \leftarrow v$, $f \leftarrow f$, $\epsilon_{mp} \leftarrow \epsilon_{mp}$
    \State $\tilde u \leftarrow u$, $\tilde v \leftarrow v$
    \State $M \leftarrow A^\top ( A U A^\top)^{-1} A $
    \State $w \leftarrow M\sqrt{U}f(v)$ \label{line:pre:w}
  \EndProcedure
  %\smallskip
  \Procedure{\textsc{Update}}{$u^{\new}, v^{\new}$}
    \LineComment{\begin{tabular}{l}The variables in this method represent the following: $u^\new,v^\new$ are the new exact values.\\
    $\tilde{u}^\new,\tilde{v}^\new$ are approximations $u^\new \approx_{\varepsilon_{mp}} \tilde{u}^\new,v^\new \approx_{\varepsilon_{mp}} \tilde{v}^\new$.\end{tabular}}
    \LineComment{\begin{tabular}{l}Vector $r$ will be the result: $r = \sqrt{\tilde{U}^\new}A^\top(A\tilde{U}^\new A^\top)^{-1} A \sqrt{\tilde{U}^\new} f(\tilde{v}^\new) $\end{tabular}}
    \LineComment{\begin{tabular}{l}For the member variables $\tilde u,\tilde v, M$ we have $w = M\sqrt{\tilde U} f(\tilde v)$ and $M = A^\top(A \tilde{U} A^\top)^{-1}A$.\\
    Note that $\tilde u,\tilde v$ are generally \emph{not} approximate versions of $u^\new, v^\new$.\end{tabular}}
    \State $y_i \leftarrow u^{\new}_i / \tilde{u}_i  - 1$, $\forall i \in [n]$
    \State Let $\pi : [n] \rightarrow [n]$ be a sorting permutation such that $|y_{\pi(i)}| \geq |y_{\pi(i+1)}|$
    \State $k \leftarrow$ the number of indices $i$ such that $|y_i| \geq \epsilon_{mp}$.
    \If {$k \ge n^a$}
      \While{$1.5 \cdot k < n$ and $|y_{\pi(1.5 k)}| \geq (1-1/\log n) |y_{\pi(k)}|$} \label{line:extra_entries}
        \State $k \leftarrow \min(\lceil 1.5 \cdot k  \rceil, n)$
      \EndWhile
    \EndIf
    \State $\tilde{u}^{\new}_{\pi(i)} \leftarrow \begin{cases} 
         u^{\new}_{\pi(i)} & i \in \{1,2,\cdots,k\} \\
         \tilde{u}_{\pi(i)} & i \in \{k+1, \cdots, n\} 
      \end{cases}$ \label{line:set_u} \Comment{$\tilde{u}^\new \approx_{\varepsilon_{mp}} u^\new$}
      \smallskip
    \State $\Delta \leftarrow \mathrm{diag}( \tilde{u}^{\new} - \tilde{u} )$ \Comment{$\Delta \in \R^{n \times n}$ and $\Delta$ has $k$ non-zero entries.}
    \State Let $S \leftarrow \pi( [k] )$ be the first $k$ indices in the permutation.
    \State Let $M_S\in \R^{n \times k}$ be the $k$ columns from $S$ of $M$.
    \State Let $M_{S,S},\Delta_{S,S} \in \R^{k \times k}$ be the $k$ rows and columns from $S$ of $M$ and $\Delta$.
    \If {$k \ge n^a$} \Comment{Perform a rank $k = 2^\ell$ update to $M$.} \label{line:entry_condition}
      \State $M \leftarrow M - M_S \cdot ( \Delta^{-1}_{S,S} + M_{S,S} )^{-1} \cdot (M_S)^\top$ \label{line:offlinewoodburry} \Comment{\begin{tabular}{r}Compute $M = A^\top ( A \tilde{U}^{\new} A^\top )^{-1} A$ via \\ Sherman-Morrison-Woodbury identity\end{tabular}}
      \vspace{-5pt}
      \textcolor{blue}{
        \State $w \leftarrow M\sqrt{\tilde U^\new} f(v^\new)$ \label{line:query1}
      }
      \State $\tilde u \leftarrow \tilde{u}^{\new}$, $\tilde{v} \leftarrow v^\new$, $ \tilde{v}^\new \leftarrow v^\new$, $r \leftarrow \sqrt{\tilde U^\new} w$
    \Else \label{line:else_branch} \Comment{This else-branch is the main difference compared to \cite{CohenLS19}.}
      \color{blue}
      \State Let $T$ be the set of indices $i$ without $(1-\varepsilon_{mp}) \tilde{v}_{i} \le v^\new_i \le (1+\varepsilon_{mp}) \tilde{v}_{i}$.
      \If{$|T| \ge n^a$} \Comment{We reset $\tilde{v} = v^\new$}
        \State $r \leftarrow \sqrt{\tilde U^\new}M\sqrt{\tilde U^\new} f(v^\new) - \sqrt{\tilde U^\new} M_S \cdot ( \Delta^{-1}_{S,S} + M_{S,S} )^{-1} \cdot (M_S)^\top \sqrt{\tilde U^\new} f(v^\new)$ \label{line:query2}
        \State $w \leftarrow M\sqrt{\tilde U} f(v^\new)$
        \State $\tilde{v} \leftarrow v^\new$, $\tilde{v}^\new \leftarrow v^\new$
      \Else
        \State $\tilde{v}^{\new}_{i} \leftarrow \begin{cases} v^{\new}_{i} & i \in T \\ \tilde{v}_{i} & i \notin T \end{cases}$ \label{line:set_v} \Comment{$\tilde{v}^\new \approx_{\varepsilon_{mp}} v^\new$}
        \State $r \leftarrow \sqrt{\tilde{U}^\new} \left( w + M(\sqrt{\tilde{U}^\new}  f(\tilde{v}^\new) - \sqrt{\tilde U}  f(\tilde{v})) - M_S \cdot ( \Delta^{-1}_{S,S} + M_{S,S} )^{-1} \cdot (M_S)^\top \sqrt{\tilde U^\new} f(\tilde{v}^\new) \right)$ \label{line:first_change_r}
		\vspace{-2pt}      
      \EndIf
      \color{black}
    \EndIf
    \LineComment{At the end of the procedure, we still have $w = M\sqrt{\tilde U} f(\tilde v)$ and $M = A^\top(A \tilde{U} A^\top)^{-1}A$}
    \LineComment{Return triple with $u^\new \approx_{\varepsilon_{mp}} \tilde{u}^\new$, $v^\new \approx_{\varepsilon_{mp}} \tilde{v}^\new$ and $r = \sqrt{\tilde{U}^\new}A^\top(A\tilde{U}^\new A)^{-1}A\sqrt{\tilde{U}^\new} f(\tilde{v}^\new)$}
    \State \Return $\tilde{u}^\new$, $\tilde{v}^\new$, $f(\tilde{v}^\new)$, $r$
  \EndProcedure
\State {\bf end datastructure}
\end{algorithmic}}
\end{algorithm*}

\subsection{Correctness}
\label{sec:projection:correctness}

The task of this subsection is to prove the following lemma, 
which says that the vectors returned by \Cref{alg:projection_maintenance}
are as specified in \Cref{lem:projection_maintenance}.
\begin{lemma}[Returned vectors of \Cref{lem:projection_maintenance}]\label{lem:correct_output}
After every update to \Cref{alg:projection_maintenance} with input $(u^\new, v^\new)$
the returned vectors $\tilde{u}^\new,\tilde{v}^\new,f(\tilde{v}^\new),r$ satisfy
$u^\new \approx_{\varepsilon_{mp}} \tilde{u}^\new$, 
$v^\new \approx_{\varepsilon_{mp}} \tilde{v}^\new$
and
$$
r = \sqrt{\tilde{U}^\new}A^\top(A\tilde{U}^\new A)^{-1}A\sqrt{\tilde{U}^\new} f(\tilde{v}^\new).
$$
\end{lemma}

Before we can prove this lemma, 
we must first prove that the internal variables of the data-structure save the correct values,
i.e. we want to prove that equation \eqref{eq:consistent_variables_intro} is correct.
For this we must first state the following lemma from {\cite{CohenLS19}, 
based on Sherman-Morison-Woodbury identity.

\ifdefined\SODAversion
\begin{lemma}\label{lem:woodbury}
\textsc{(\cite{CohenLS19}, based on Sherman-Morison-Woodbury identity)}
If $M = A^\top(A\tilde{U}A^\top)^{-1}A$ at the start of the update of \Cref{alg:projection_maintenance} and $M_S, M_{S,S},\Delta_{S,S}$ are chosen as described in \Cref{alg:projection_maintenance}, then we have
$$M - M_S \cdot ( \Delta^{-1}_{S,S} + M_{S,S} )^{-1} \cdot (M_S)^\top = A^\top (A \tilde{U}^\new A^\top)^{-1} A$$
\end{lemma}
\else
\begin{lemma}[{\cite{CohenLS19}, based on Sherman-Morison-Woodbury identity}]\label{lem:woodbury}
If $M = A^\top(A\tilde{U}A^\top)^{-1}A$ at the start of the update of \Cref{alg:projection_maintenance} and $M_S, M_{S,S},\Delta_{S,S}$ are chosen as described in \Cref{alg:projection_maintenance}, then we have
$$M - M_S \cdot ( \Delta^{-1}_{S,S} + M_{S,S} )^{-1} \cdot (M_S)^\top = A^\top (A \tilde{U}^\new A^\top)^{-1} A$$
\end{lemma}
\fi

We can now prove that the internal variables store the correct values.

\begin{lemma}\label{lem:consistent_variables}
At the start of every update to \Cref{alg:projection_maintenance} we have
\begin{align}
w = M\sqrt{\tilde{U}} f(\tilde{v})
\hspace{25pt} \text{and} \hspace{25pt}
M = A^\top(A\tilde{U} A)^{-1}A.
\label{eq:consistent_variables}
\end{align}
\end{lemma}

\begin{proof}
If it is the first update after the initialization, 
then the claim is true by definition of the procedure \textsc{Initialize}.
Next, we prove that at the end of every call to \textsc{Update} 
we satisfy \eqref{eq:consistent_variables}, 
if \eqref{eq:consistent_variables} was satisfied 
at the start of \textsc{Update}. 
This then implies \Cref{lem:consistent_variables}.
If $k \ge n^a$, then \cref{line:offlinewoodburry} makes sure that 
$M = A^\top(A\tilde{U}^\new A^\top)^{-1}A$ (see \Cref{lem:woodbury}).
The next lines set 
$w \leftarrow M \sqrt{\tilde{U}^\new} f(v^\new)$,
$\tilde{v} \leftarrow v^\new$ and $\tilde{u} \leftarrow \tilde{u}^\new$.
Thus \eqref{eq:consistent_variables} is satisfied for the case $k \ge n^a$.
If $|T| \ge n^a$, then we compute 
$w \leftarrow M \sqrt{\tilde{U}} f(v^\new)$ 
and set $\tilde{v} \leftarrow v^\new$. 
The matrices $M$ and $\tilde{U}$ are not modified, 
so \eqref{eq:consistent_variables} is satisfied.
If $|T| < n^a$, then we do not change $M,\tilde{u},\tilde{v}$ or $r$, 
so \eqref{eq:consistent_variables} is satisfied.

\end{proof}

We now prove the correctness of \Cref{alg:projection_maintenance} by proving \Cref{lem:correct_output}.

\begin{proof}[Proof of \Cref{lem:correct_output}]

Note that we always have 
$u^\new \approx_{\varepsilon_{mp}} \tilde{u}^\new$ 
by \cref{line:set_u}.

\paragraph{Case $k \ge n^a$:}
In \cref{line:offlinewoodburry} 
we have set $M = A^\top(A\tilde{U}^\new A^\top)^{-1}A$ 
(see \Cref{lem:woodbury,lem:consistent_variables}).
Hence by setting $r \leftarrow \sqrt{\tilde{U}^\new}w = \sqrt{\tilde{U}^\new}M\sqrt{\tilde{U}^\new} f(v^\new)$, 
and $\tilde{v}^\new \leftarrow v^\new$,
all claims of \Cref{lem:correct_output} are satisfied.

\paragraph{Case $|T| \ge n^a$:}
In this case we set $r$ to the following value:
\ifdefined\SODAversion
\begin{align*}
r
\leftarrow&~ \sqrt{\tilde U^\new}M\sqrt{\tilde U^\new} f(v^\new)\\
&~ - \sqrt{\tilde U^\new} M_S \cdot ( \Delta^{-1}_{S,S} + M_{S,S} )^{-1}\\
&~ \cdot (M_S)^\top \sqrt{\tilde U^\new} f(v^\new) \\
%& = \sqrt{\tilde U^\new}(M - \sqrt{\tilde U^\new} M_S \cdot ( \Delta^{-1}_{S,S} + M_{S,S} )^{-1} \cdot (M_S)^\top) \sqrt{\tilde U^\new} f(v^\new) \\
=&~ \sqrt{\tilde U^\new}(A^\top(A \tilde{U}^\new A^\top)^{-1}A) \sqrt{\tilde U^\new} f(v^\new)
\end{align*}
\else
\begin{align*}
r
&\leftarrow \sqrt{\tilde U^\new}M\sqrt{\tilde U^\new} f(v^\new) - \sqrt{\tilde U^\new} M_S \cdot ( \Delta^{-1}_{S,S} + M_{S,S} )^{-1} \cdot (M_S)^\top \sqrt{\tilde U^\new} f(v^\new) \\
%& = \sqrt{\tilde U^\new}(M - \sqrt{\tilde U^\new} M_S \cdot ( \Delta^{-1}_{S,S} + M_{S,S} )^{-1} \cdot (M_S)^\top) \sqrt{\tilde U^\new} f(v^\new) \\
& = \sqrt{\tilde U^\new}(A^\top(A \tilde{U}^\new A^\top)^{-1}A) \sqrt{\tilde U^\new} f(v^\new)
\end{align*}
\fi
Here the equality comes from \Cref{lem:woodbury,lem:consistent_variables}. 
Further, we set $\tilde{v}^\new \leftarrow v^\new$, 
so \Cref{lem:correct_output} is correct for the case $|T| \ge n^a$.

\paragraph{Case $|T| < n^a$:} 
Here $v^\new \approx_{\varepsilon_{mp}} \tilde{v}^\new$ by \cref{line:set_v},
so we are left with verifying $r$.
First note that 
$w = M\sqrt{\tilde U}f(\tilde{v})$ by \Cref{lem:consistent_variables}, 
so $w+M(\sqrt{\tilde{U}^\new}f(\tilde{v}^\new) - \sqrt{\tilde U} f(\tilde{v}))
= M\sqrt{\tilde{U}^\new}f(\tilde{v}^\new)$.
Thus $r$ is set to the following term:
\ifdefined\SODAversion
\begin{align*}
r \leftarrow&~
\sqrt{\tilde{U}^\new} \left( w + M(\sqrt{\tilde{U}^\new}  f(\tilde{v}^\new) - \sqrt{\tilde U}  f(\tilde{v}))\right. \\
&~ \left.- M_S \cdot ( \Delta^{-1}_{S,S} + M_{S,S} )^{-1} \cdot (M_S)^\top \sqrt{\tilde U^\new} f(\tilde{v}^\new) \right)\\
=&~
\sqrt{\tilde{U}^\new} \left( M \sqrt{\tilde U^\new} f(\tilde{v}^\new)\right. \\
&~ \left.- M_S \cdot ( \Delta^{-1}_{S,S} + M_{S,S} )^{-1} \cdot (M_S)^\top \sqrt{\tilde U^\new} f(\tilde{v}^\new) \right)\\
=&~\text{[continued on next page]}
\end{align*}
\begin{align*}
=&~
\sqrt{\tilde{U}^\new} \left( M  - M_S \cdot ( \Delta^{-1}_{S,S} + M_{S,S} )^{-1} \cdot (M_S)^\top \right) \\
&~ \cdot \sqrt{\tilde U^\new} f(\tilde{v}^\new)\\
=&~
\sqrt{\tilde{U}^\new} A^\top(A\tilde{U}^\new A^\top)^{-1}  \sqrt{\tilde U^\new} f(\tilde{v}^\new)
\end{align*}
\else
\begin{align*}
r &\leftarrow
\sqrt{\tilde{U}^\new} \left( w + M(\sqrt{\tilde{U}^\new}  f(\tilde{v}^\new) - \sqrt{\tilde U}  f(\tilde{v})) 
- M_S \cdot ( \Delta^{-1}_{S,S} + M_{S,S} )^{-1} \cdot (M_S)^\top \sqrt{\tilde U^\new} f(\tilde{v}^\new) \right)\\
&=
\sqrt{\tilde{U}^\new} \left( M \sqrt{\tilde U^\new} f(\tilde{v}^\new) - M_S \cdot ( \Delta^{-1}_{S,S} + M_{S,S} )^{-1} \cdot (M_S)^\top \sqrt{\tilde U^\new} f(\tilde{v}^\new) \right)\\
&=
\sqrt{\tilde{U}^\new} \left( M  - M_S \cdot ( \Delta^{-1}_{S,S} + M_{S,S} )^{-1} \cdot (M_S)^\top \right)  \sqrt{\tilde U^\new} f(\tilde{v}^\new)\\
&=
\sqrt{\tilde{U}^\new} A^\top(A\tilde{U}^\new A^\top)^{-1}  \sqrt{\tilde U^\new} f(\tilde{v}^\new)
\end{align*}
\fi
Where the last equality comes from \Cref{lem:woodbury,lem:consistent_variables}.

\end{proof}

\subsection{Complexity}
\label{sec:projection:complexity}

In this section we will bound the complexity of \Cref{alg:projection_maintenance},
proving the stated complexity bound in \Cref{lem:projection_maintenance}:

\begin{lemma}[Complexity bound of \Cref{lem:projection_maintenance}]\label{lem:projection_maintenance_complexity}
If the updates to \Cref{alg:projection_maintenance} 
satisfy the condition \eqref{eq:small_change\theimportcounter}
as stated in \Cref{lem:projection_maintenance},
then after $T$ updates the total update time of \Cref{alg:projection_maintenance} is
$$
O\left(T \cdot \left(C / \varepsilon_{mp} (n^{\omega-1/2} + n^{2-a/2+o(1)}) \log n + n^{1+a}\right)\right).
$$
The preprocessing requires $O(n^2d^{\omega-2})$ time.
\end{lemma}
As our data-structure is a modification of the data-structure presented in \cite{CohenLS19},
we must only analyze the complexity of the modified part. 
To bound the complexity of the unmodified sections of our algorithm, 
\ifdefined\SODAversion
we refer to \cite{CohenLS19} and the full version of this paper \cite{Brand19}.
The complexity analysis in \cite{CohenLS19} requires an entire section (about 7 pages) via analysis of some complicated potential function.
In our full version, we present an alternative simpler proof.
\else
we will here refer to \cite{CohenLS19}.
The complexity analysis in \cite{CohenLS19} requires an entire section (about 7 pages) via analysis of some complicated potential function.
In the Appendix (\Cref{lem:appendix:cohens_bound}) we present an alternative simpler proof.
\fi

\ifdefined\SODAversion
\begin{lemma}[\cite{CohenLS19,Brand19}]\label{lem:cohens_bound}
\else
\begin{lemma}[{\cite{CohenLS19}, alternatively \Cref{lem:appendix:cohens_bound}}]\label{lem:cohens_bound}
\fi
The preprocessing requires $O(n^2d^{\omega-2})$ time.
After $T$ updates the total time of all updates of \Cref{alg:projection_maintenance}, 
when ignoring the branch for $k < n^a$ (so we assume that branch of \cref{line:else_branch} has cost $0$), 
is
$$
O(T \cdot C / \varepsilon_{mp} (n^{\omega-1/2} + n^{2-a/2+o(1)}) \log n).
$$
\end{lemma}

\begin{proof}
When ignoring the branch of \cref{line:else_branch}, 
then our algorithm performs the same operations as \cite{CohenLS19}[Algorithm 3]
and we both maintain $M$ in the exact same way.
%(For completeness sake we state their algorithm in \Cref{app:original_algorithm}.)
The only difference is that we also compute the vector $r$ in \cref{line:query1}, 
but this requires only $O(n^2)$ time and is subsumed 
by the complexity of \cref{line:offlinewoodburry}.
Thus our time complexity (when ignoring the branch of \cref{line:else_branch})
can be bounded by the update complexity of \cite{CohenLS19}[Algorithm 3], 
which is the complexity stated in \Cref{lem:cohens_bound}.
In the same fashion we can bound the complexity of the preprocessing.
The preprocessing of \cite{CohenLS19}[Algorithm 3] takes $O(n^2d^{\omega-2})$ time, 
where their algorithm computes only the matrix $M$. 
The only difference in our algorithm is 
that we also compute the vector $w$ in \cref{line:pre:w}. 
The required $O(n^2)$ time to compute $w$ is subsumed by computing $M$.
\end{proof}

\begin{proof}[Proof of \Cref{lem:projection_maintenance_complexity}]
In order to prove \Cref{lem:projection_maintenance_complexity} 
we only need to bound the complexity of the branch for the case $k < n^a$. 
The time required by all other steps of \Cref{alg:projection_maintenance} 
is already bounded by \Cref{lem:cohens_bound}.

In every update we must compute $(\Delta_{S,S}^{-1}+M_{S,S})^{-1}$, 
which takes $O(n^{a\cdot\omega})$ time via the assumption $k < n^a$.
Additionally, if $|T| < n^a$, then one update requires 
additional $O(n^{1+a})$ operations to compute $r$ and $w$, 
because $(f(\tilde{v}^\new) - f(\tilde{v}))$ and $(\sqrt{\tilde{U}} - \sqrt{\tilde{U}^\new})$ 
both have at most $n^a$ non-zero entries 
and $M_S$ is a $n \times n^a$ matrix.

If $T \ge n^a$, then computing $r$ and $w$ can take up to $O(n^2)$ operations. 
This can happen at most every $O(n^{a/2}\varepsilon_{mp}/C)$ updates by \Cref{lem:error_after_k_rounds},
because $\sum_{i=1}^n ((v_i^\new - v_i)/ v_i)^2 \le C^2$, 
Hence the amortized time per update is $O(n^{2-a/2}C/\varepsilon_{mp})$.

Note that by assuming $a \le \alpha \le 1$ the term $O(n^{a\cdot\omega})$ is subsumed by $O(n^{1+a})$,
because $\omega \le 3 - \alpha$, so $a \cdot \omega \le a(3-\alpha) \le a(3-a) \le 1 + a$.
\end{proof}

This concludes the proof of \Cref{lem:projection_maintenance}.

\section{Central Path Method}
\label{sec:central_path}

In this section we prove the main result \Cref{thm:LP},
by showing how to use the projection maintenance algorithm of \Cref{sec:projection_maintenance} 
to obtain a fast deterministic algorithm for solving linear programs.

The algorithm for \Cref{thm:LP} is based on 
the short step central path method, outlined in \Cref{sec:outline:central_path}:
We construct some feasible solution triple $(x,y,s)$ with $xs =: \mu \approx 1$
and then repeatedly decrease $t$ while maintaining $x,s$ such that $\mu$ stays close to $t$.
Once $t$ is small enough, we have a good approximate solution.
This is a high-level summary of \Cref{alg:main}, 
which first constructs a solution, and then runs a \textsc{while}-loop until $t$ is small enough.
The actual hard part, maintaining the solution pair $x,s$ 
with $\mu \approx t$, is done in \Cref{alg:approximate_step}.
For this task, \Cref{alg:approximate_step} solves 
a linear system (similar to \eqref{eq:interior_point_system} in \Cref{sec:outline:central_path})
via the data-structure of \Cref{lem:projection_maintenance}.
The majority of this section is dedicated to proving 
that \Cref{alg:approximate_step} does not require too much time 
and does indeed maintain the solution pairs $(x,s)$
with $\mu \approx t$.
For this we must verify the following three properties:
\begin{itemize}
\item \Cref{alg:approximate_step} does solve an approximate variant of the linear system \eqref{eq:interior_point_system}.
\item We do not change the linear system too much between two calls to \Cref{alg:approximate_step}.
Otherwise the data-structure of \Cref{lem:projection_maintenance} would become too slow.
\item The approximate result obtained in \Cref{alg:approximate_step}
is good enough to maintain $x,s$ such that $\mu$ is close to $t$.
\end{itemize}

The proof for this is based on 
the \emph{stochastic central path method} 
by Cohen et al.~\cite{CohenLS19}.
In \cite{CohenLS19}, they randomly sampled a certain vector, 
while in our algorithm this vector will be approximated deterministically 
via the data-structure of \Cref{alg:projection_maintenance}.
This derandomization has the nice side-effect, 
that we can skip many steps of Cohen et al.'s proof.
For example they had to bound the variance of random vectors, 
which is no longer necessary for our algorithm.

The outline of this section is as follows. 
We first explain in more detail 
how \Cref{alg:approximate_step} works in \Cref{sec:central_path:algorithm}, 
where we also verify the first requirement,
that \Cref{alg:approximate_step} does indeed solve the system \eqref{eq:interior_point_system} approximately.
In the next \Cref{sub:bound_change}, we check that the input parameters 
for the data-structure of \Cref{lem:projection_maintenance}, 
used by \Cref{alg:approximate_step}, do not change too much per iterations.
The last \Cref{sub:mu_close_to_t} verifies, 
that we indeed always have $\mu \approx t$. 
We also consolidate all results in the last subsection by proving the main result \Cref{thm:LP}.

\begin{algorithm*}[t]\caption{Iterative loop of the central path method}\label{alg:main}
\begin{algorithmic}[1]
\Procedure{\textsc{Main}}{$A,b,c,\delta$} \Comment{\Cref{thm:LP}}
	\State $\varepsilon \leftarrow 1/(1500 \ln n)$ \Comment{Step size. Controls how much we decrease $t$ in each iteration.}
	\State $\varepsilon_{mp} \leftarrow 1/(1500 \ln n)$ \Comment{Accuracy parameter for \Cref{alg:approximate_step,alg:projection_maintenance}}
	\State $\lambda \leftarrow 40 \ln n$ \Comment{Parameter for the potential function in \Cref{alg:approximate_step}.}
	\State $t \leftarrow 1$ \Comment{Measures the progress so far.}	
	\State Modify the linear program according to $\Cref{lem:feasible_LP}$ for $\gamma = \min\{ \delta, 1/\lambda \}$ 
	and obtain an initial $x$ and $s$.
	\State $\textsc{InitializeApproximateStep}(A,x,s,t,\lambda,\varepsilon_{mp})$ \Comment{Initialize \Cref{alg:approximate_step}}
	\While{$t > \delta^2/(2n)$} \Comment{We stop once the precision is good}
		\LineComment{Decrease $t$ to $t^\new$ and find new $x^\new,s^\new$ such that $x^\new s^\new =:\mu^\new \approx_{0.1} t^\new$}
		\State $(x^{\new} , s^{\new}, t^\new) \leftarrow \textsc{ApproximateStep}(x,s,t,\varepsilon)$
		\State $(x,s) \leftarrow (x^{\new},s^{\new})$, $t \leftarrow t^{\new}$ 
	\EndWhile
    \State Use $\Cref{lem:feasible_LP}$ to transform $x$ to an approximate solution of the original linear program.
\EndProcedure
\end{algorithmic}
\end{algorithm*}

\subsection{Using the Projection Maintenance Data-Structure in the Central Path Method}
\label{sec:central_path:algorithm}

In this section we outline how \Cref{alg:approximate_step} works
and we prove that it does indeed solve 
the linear system \eqref{eq:interior_point_system} 
(outlined in \Cref{sec:outline:central_path}) 
in some approximate way.
The high-level idea of \Cref{alg:approximate_step} is as follows:
In order to maintain $\mu$ close to $t$, 
we want to measure the distance via some potential function $\Phi(\mu/t-1)$.
As we want to minimize the distance, 
it makes sense to change $\mu$ by some $\delta_\mu$, 
which points in the same direction as $-\nabla \Phi(\mu/t-1)$.
We can find out how to change $x$ and $s$, 
in order to change $\mu$ by approximately $\delta_\mu$, 
by solving the linear system \eqref{eq:interior_point_system} 
via the data-structure of \Cref{lem:projection_maintenance}.

In reality, we choose $\delta_\mu$ to be slightly different:
$$\delta_\mu := 
( \frac{t^{\new}}{t} - 1 ) \mu 
- \frac{\varepsilon}{2} \cdot t^{\new} \cdot 
\frac{ \nabla \Phi( \mu / t - 1 ) }
{ \| \nabla \Phi(  \mu / t  - 1 ) \|_2 },$$
where $t^\new := (1-\frac{\varepsilon}{3\sqrt{n}})$ is the new smaller value that we want to set $t$ to, 
and $\varepsilon$ is a parameter for how large our step size should be for decreasing $t$.

This choice of $\delta_\mu$ is motivated by the fact, 
that the first term $( \frac{t^{\new}}{t} - 1 ) \mu $ 
leads to some helpful cancellations in later proofs. 
The second term is the one pointing in the direction of $-\nabla \Phi(\mu/t-1)$, 
which is motivated by decreasing $\Phi(\mu/t-1)$.

\begin{algorithm*}[t]\caption{\textsc{ApproximateStep} For the given solution pair $x,s$ move $xs \approx t$ closer to $t^\new$}\label{alg:approximate_step}
\begin{algorithmic}[1]
\Variables{}
\State $\text{mp}_{\sqrt{\mu}}$, $\text{mp}_{\nabla \Phi}$
\EndVariables
\\
\Procedure{\textsc{InitializeApproximateStep}}{$A,x,s,t,\lambda,\varepsilon_{mp}$}
	\State $u \leftarrow \frac{x}{s}$, \hspace{10pt} $\mu \leftarrow xs$
	\LineComment{Maintains approximation of 
		$\sqrt{U}A^\top(A U A^\top)^{-1} A \sqrt{U} \sqrt{\mu}$ 
		via \Cref{alg:projection_maintenance}.}
	\State $\text{mp}_{\sqrt{\mu}}$.\textsc{Initialize}$(A,u,x \mapsto \sqrt{x}, \mu, \varepsilon_{mp},)$
	\LineComment{Maintains approximation of 
		$\sqrt{U}A^\top(A U A^\top)^{-1} A \sqrt{U} \frac{\nabla \Phi_\lambda(\mu/t-1)}{\sqrt{\mu/t}}$ 
		via \Cref{alg:projection_maintenance}.}
	\State $\text{mp}_{\nabla\Phi}$.\textsc{Initialize}$(A,u,x \mapsto \lambda \sinh(\lambda (x - 1))/\sqrt{x}, \mu/t, \varepsilon_{mp})$
\EndProcedure
\\
\Procedure{\textsc{ApproximateStep}}{$x,s,t,\varepsilon$}\\
	\Comment{One step of the modified short step central path method}
	\State $t^\new \leftarrow (1-\frac{\varepsilon}{3\sqrt{n}}) t$, \hspace{10pt} $\mu \leftarrow xs$
	\State $(\tilde{u}, \cdot, v, p_v) \leftarrow 
	\text{mp}_{\sqrt{\mu}}.\textsc{Update}(u,\mu)$, \hspace{10pt}
	$(\tilde{u}, m, w, p_w) \leftarrow 
	\text{mp}_{\nabla\Phi}.\textsc{Update}(u,\mu/t)$
	\label{line:projection}	
		\LineComment{Note that both instances always receive the same $u$, so they also return the same $\tilde{u}$.}
	\State $\tilde{\mu} \leftarrow mt$
	\State $\tilde{x} \leftarrow x\sqrt{\frac{\tilde{\mu}}{\mu}\frac{\tilde{u}}{u}}$,
	\hspace{10pt} $\tilde{s} \leftarrow s\sqrt{\frac{\tilde{\mu}}{\mu}\frac{u}{\tilde{u}}}$
	\Comment{Thus $\tilde{x}\tilde{s} = \tilde{\mu}$ and $\tilde{x}/\tilde{s} = \tilde{u}$}
	\State $\tdelta_t \leftarrow ( \frac{t^{\new}}{t} - 1 ) v \sqrt{\tilde{\mu}}$, \hspace{10pt}
	$\tdelta_\Phi \leftarrow - \frac{\varepsilon}{2} \cdot t^{\new} \cdot \frac{\sqrt{\tilde{\mu}/t} \: w}{ \| \nabla \Phi_\lambda(\tilde{\mu}/t -1) \|_2 }$, \hspace{10pt}
	$\tdelta_\mu \leftarrow \tdelta_t + \tdelta_\Phi$
	\State $p \leftarrow ( \frac{t^{\new}}{t} - 1 )p_v - \frac{\varepsilon}{2} \cdot t^{\new} \cdot \frac{p_w}{\sqrt{t} \| \nabla \Phi_\lambda(\tilde{\mu}/t -1) \|_2 }$

	\State $\tilde{\delta}_s \leftarrow \frac{ \tilde{s} }{ \sqrt{\tilde{\mu} } } p$, \hspace{10pt} 
	$\tilde{\delta}_x \leftarrow \frac{1}{\tilde{s}} \tilde{\delta}_{\mu} - \frac{\tilde{x}}{\sqrt{\tilde{\mu}}} p$
	
	\State \Return $( x + \tilde{\delta}_x  ,  s + \tilde{\delta}_s, t^\new)$
\EndProcedure
\end{algorithmic}
\end{algorithm*}

\paragraph{Maintaining approximate solutions}

One can split $\delta_\mu$ into the two terms 
$\delta_t = ( \frac{t^{\new}}{t} - 1 ) \mu $ and
$\delta_\Phi = \frac{\varepsilon}{2} \cdot t^{\new} \cdot 
\frac{ \nabla \Phi( \mu / t - 1 ) }
{ \| \nabla \Phi(  \mu / t  - 1 ) \|_2 }$.

\Cref{alg:approximate_step} approximates both vectors in a different way. 
Specifically, given $x,s,\mu,\delta_t,\delta_\Phi,\delta_\mu$, 
\Cref{alg:approximate_step} internally maintains approximations 
$\tilde{x},\tilde{s},\tilde{\mu},\tdelta_t,\tdelta_\Phi,\tdelta_\mu$ with the following properties 
(here $\varepsilon_{mp} > 0$ is the accuracy parameter for \Cref{lem:projection_maintenance})
\begin{align}
x &\approx_{\varepsilon_{mp}} \tilde{x}, &&
s \approx_{2\varepsilon_{mp}} \tilde{s} \notag \\
xs = \mu &\approx_{\varepsilon_{mp}} \tilde{\mu} = \tilde{x}\tilde{s}, &&
\tdelta_\Phi = \frac{\varepsilon}{2} \cdot t^{\new} \cdot 
\frac{ \nabla \Phi( \tilde\mu / t - 1 ) }
{ \| \nabla \Phi( \tilde\mu / t  - 1 ) \|_2 }, \label{eq:approximations} \\
\delta_t &\approx_{\varepsilon_{mp}} \tdelta_t, &&
\tdelta_\mu = \tdelta_t + \tdelta_\Phi  \notag
\end{align}
and for these approximate values, we solve the following system 
(which is the same as \eqref{eq:interior_point_system}, but using the approximate values):
\begin{align} 
	\tilde{X} \tilde{\delta}_s + \tilde{S} \tilde{\delta}_x = & ~ \tilde{\delta}_{\mu},
	\label{eq:interior_point_system_approximate}\\
	A \tilde{\delta}_x = & ~ 0, \notag \\
	A^\top \tilde{\delta}_y + \tilde{\delta}_s = & ~  0. \notag
\end{align}

We prove in two steps that \Cref{alg:approximate_step} 
does indeed solve \eqref{eq:interior_point_system_approximate}
for approximate values as in \eqref{eq:approximations}:
First we prove in \Cref{lem:error_delta} that the approximations 
are as stated in \eqref{eq:approximations}, 
then we show in \Cref{lem:approximate_solution} 
that we indeed solve the linear system \eqref{eq:interior_point_system_approximate}.

Note that $\tdelta_\Phi$ is not an approximation 
of $\delta_\Phi$ in the classical sense 
(likewise $\tdelta_\mu$ and $\delta_\mu$) 
and the vectors could point in completely different directions. 
They are only ``approximate" in the sense 
that their definition is the same, 
but for $\tdelta_\Phi$ we replace $\mu$ 
by the approximate $\tilde{\mu}$.

As outlined in the overview \Cref{sec:outline:modified_central_path},
this results in our algorithm not always decreasing 
the difference between $\mu$ and $t$.
We prove in \Cref{sub:mu_close_to_t}
that this is not a problem, 
if we use the following potential function $\Phi$,
accuracy parameter $\varepsilon_{mp}$ (for \Cref{lem:projection_maintenance})
and step size $\varepsilon$.
\begin{Definition}
$$\Phi_\lambda(x) := \sum_{i=1}^n \cosh(\lambda x_i),$$
where $\cosh(x) := (e^x + e^{-x})/2$, $\lambda = 40 \ln n$.\\
For the step size $\varepsilon$ and the accuracy parameter $\varepsilon_{mp}$ for \Cref{lem:projection_maintenance}, 
assume $0 < \varepsilon_{mp} \le \varepsilon \le 1/(1500 \ln n)$.
\end{Definition}

\begin{lemma}\label{lem:error_delta}
The computed vectors $\tilde{x},\tilde{s},\tilde{\mu},\tdelta_t,\tdelta_\Phi,\tdelta_\mu$ 
in \Cref{alg:approximate_step} satisfy the following properties:
Let $\tilde{\mu}/t$ be the approximation of $\mu/t$ maintained internally by $\text{mp}_\Phi$, 
then $\mu \approx_{\varepsilon_{mp}} \tilde{\mu}$ and 
$\tdelta_\Phi = -\frac{\varepsilon}{2} \cdot t^\new \cdot 
\frac{\nabla \Phi_\lambda(\tilde{\mu}/t-1)}
{\|\nabla \Phi_\lambda(\tilde{\mu}/t-1)\|}$.
Further $\delta_t \approx_{\varepsilon_{mp}} \tdelta_t$,
$x \approx_{\varepsilon_{mp}} \tilde{x}$, 
$s \approx_{2\varepsilon_{mp}} \tilde{s}$.
\end{lemma}

\begin{proof}
The returned vector $m$ in line \ref{line:projection} 
is an approximation in the sense that 
$\mu/t \approx_{\varepsilon_{mp}} m$, 
which means $\mu \approx_{\varepsilon_{mp}} mt =: \tilde{\mu}$.
We have $\frac{x}{s} =: u \approx_{\varepsilon_{mp}} \tilde{u}$, 
hence $\frac{u}{\tilde{u}} \approx_{\varepsilon_{mp}} 1$ 
and $1 \approx_{\varepsilon_{mp}} \frac{\tilde{u}}{u}$.
Thus $x \approx_{\varepsilon_{mp}} x \sqrt{\frac{\tilde{\mu}}{\mu}\frac{\tilde{u}}{u}} = \tilde{x}$
and $s \approx_{2\varepsilon_{mp}} s \sqrt{\frac{\tilde{\mu}}{\mu}\frac{u}{\tilde{u}}} = \tilde{s}$.

As potential function we have chosen 
$\Phi_\lambda(x) = \sum_{i=1}^n \cosh(x_i)$, 
so 
$(\nabla \Phi_\lambda(x-1)/\sqrt{x})_i = \lambda \sinh( \lambda(x_i-1)) / \sqrt{x_i}.$
This means $\lambda \sinh( \lambda(x-1)) / \sqrt{x}$ for $x = \mu/t$ is $\nabla \Phi_\lambda(\mu/t-1)/\sqrt{\mu/t}$
and 
$w = \lambda \sinh( \lambda(m-1)) / \sqrt{m} = \nabla \Phi_\lambda(\tilde{\mu}/t-1)/\sqrt{\tilde{\mu}/t}.$
Hence we have that
$\tdelta_\Phi = -\frac{\varepsilon}{2} \cdot t^\new \cdot \frac{\sqrt{\tilde{\mu}/t}w}{\|\nabla \Phi_\lambda(\tilde{\mu}/t-1)\|_2}
= -\frac{\varepsilon}{2} \cdot t^\new \cdot \frac{\nabla \Phi_\lambda(\tilde{\mu}/t-1)}{\|\nabla \Phi_\lambda(\tilde{\mu}/t-1)\|_2}
$.
We also have
$\delta_t = ( \frac{t^{\new}}{t} - 1 ) \mu$, 
$\mu \approx_{\varepsilon_{mp}} v^2$ and 
$\mu \approx_{\varepsilon_{mp}} \tilde{\mu}$, 
so $\mu \approx_{\varepsilon_{mp}} v \sqrt{\tilde{\mu}}$ 
which implies 
$\delta_t \approx_{\varepsilon_{mp}} ( \frac{t^{\new}}{t} - 1 ) v \sqrt{\tilde{\mu}} =: \tdelta_t$.

\end{proof}

\begin{lemma}
\label{lem:approximate_solution}
The computed vectors in \Cref{alg:approximate_step} satisfy the following linear system:
\begin{align*} 
	\tilde{X} \tilde{\delta}_s + \tilde{S} \tilde{\delta}_x = & ~ \tilde{\delta}_{\mu}, \\
	A \tilde{\delta}_x = & ~ 0%, \\
%	A^\top \tilde{\delta}_y + \tilde{\delta}_s = & ~  0.
\end{align*}
\end{lemma}

\begin{proof}
We define the following projection matrix:
\ifdefined\SODAversion
\begin{align*}
P &:= \sqrt{\tilde{U}} A^\top (A \tilde{U} A^\top)^{-1} A \sqrt{\tilde{U}} \\
&= \sqrt{\tilde{X}/\tilde{S}} A^\top (A (\tilde{X}/\tilde{S}) A^\top)^{-1} A \sqrt{\tilde{X}/\tilde{S}}
\end{align*}
Then we have
\begin{align*}
&~ p_v (\frac{t^\new}{t}-1)
= P v (\frac{t^\new}{t}-1) \\
=&~ P \frac{1}{\sqrt{\tilde{x}\tilde{s}}} \sqrt{\tilde{\mu}} \: v (\frac{t^\new}{t}-1)
= P \frac{1}{\sqrt{\tilde{x}\tilde{s}}} \tdelta_t
\end{align*}
and
\begin{align*}
&~ -\frac{\varepsilon}{2} \cdot \frac{t^\new p_w}{\sqrt{t} \|\nabla \Phi_\lambda(\tilde{\mu}/t-1)\|_2} \\
=&~ -\frac{\varepsilon}{2} \cdot \frac{t^\new P w}{\sqrt{t}\|\nabla \Phi_\lambda(\tilde{\mu}/t-1)\|_2} \\
=&~ -\frac{\varepsilon}{2} \cdot \frac{t^\new P \frac{1}{\sqrt{\tilde{x}\tilde{s}}} \sqrt{\tilde{\mu}} \: w}{\sqrt{t}\|\nabla \Phi_\lambda(\tilde{\mu}/t-1)\|_2} 
= P \frac{1}{\sqrt{\tilde{x}\tilde{s}}} \tdelta_\Phi
.
\end{align*}
\else
$$P := \sqrt{\tilde{U}} A^\top (A \tilde{U} A^\top)^{-1} A \sqrt{\tilde{U}}
= \sqrt{\tilde{X}/\tilde{S}} A^\top (A (\tilde{X}/\tilde{S}) A^\top)^{-1} A \sqrt{\tilde{X}/\tilde{S}}$$
Then we have
$$
p_v (\frac{t^\new}{t}-1)
= P v (\frac{t^\new}{t}-1)
= P \frac{1}{\sqrt{\tilde{x}\tilde{s}}} \sqrt{\tilde{\mu}} \: v (\frac{t^\new}{t}-1)
= P \frac{1}{\sqrt{\tilde{x}\tilde{s}}} \tdelta_t
$$
and
$$
-\frac{\varepsilon}{2} \cdot \frac{t^\new p_w}{\sqrt{t} \|\nabla \Phi_\lambda(\tilde{\mu}/t-1)\|_2}
= -\frac{\varepsilon}{2} \cdot \frac{t^\new P w}{\sqrt{t}\|\nabla \Phi_\lambda(\tilde{\mu}/t-1)\|_2}
= -\frac{\varepsilon}{2} \cdot \frac{t^\new P \frac{1}{\sqrt{\tilde{x}\tilde{s}}} \sqrt{\tilde{\mu}} \: w}{\sqrt{t}\|\nabla \Phi_\lambda(\tilde{\mu}/t-1)\|_2} 
= P \frac{1}{\sqrt{\tilde{x}\tilde{s}}} \tdelta_\Phi
.$$
\fi
Hence the change to $x$ and $s$ is given by
$$
\tdelta_s = \frac{\tilde{s}}{\sqrt{\tilde{x}\tilde{s}}} p
= \frac{\tilde{s}}{\sqrt{\tilde{x}\tilde{s}}} P \frac{1}{\sqrt{\tilde{x}\tilde{s}}} (\tdelta_t + \tdelta_\Phi)
= \frac{\tilde{s}}{\sqrt{\tilde{x}\tilde{s}}} P \frac{1}{\sqrt{\tilde{x}\tilde{s}}} \tdelta_\mu
$$

$$
\tdelta_x
= \frac{1}{\tilde{s}} \tdelta_\mu - \frac{\tilde{x}}{\sqrt{\tilde{x}\tilde{s}}} P \frac{1}{\sqrt{\tilde{x}\tilde{s}}} p
= \frac{\tilde{x}}{\sqrt{\tilde{x}\tilde{s}}} (\I - P) \frac{1}{\sqrt{\tilde{x}\tilde{s}}} \tdelta_\mu
$$
\Cref{lem:approximate_solution} is thus given by \Cref{lem:solution_LP_system}.
\end{proof}

\subsection{Bounding the change per iteration}
\label{sub:bound_change}

\Cref{alg:approximate_step} uses the data-structure of \Cref{lem:projection_maintenance}.
The complexity of this data-structure depends on how much the input parameters 
(in our case $u:=x/s$, $\mu$ and $\mu/t$) change per iteration.
In this section we prove:
\begin{lemma}\label{lem:change_mu_u}
Assume $\mu \approx_{0.1} t$.
Let $\mu^\new := (x+\tdelta_x)(s+\tdelta_s)$, 
the value of $\mu$ in the upcoming iteration,
and let $u := \frac{x}{s}$, $u^\new := \frac{x+\tdelta_x}{s+\tdelta_s}$, then
\ifdefined\SODAversion
\begin{align*}
\|\mu^{-1}(\mu^\new - \mu) \|_2 &\le 2.5\varepsilon,\\
\|(\mu/t)^{-1}(\mu^\new/t^\new - \mu/t) \|_2 &\le 3\varepsilon,\\
\| (u^\new - u) / u \|_2 &\le 3 \varepsilon.
\end{align*}
\else
$$\|\mu^{-1}(\mu^\new - \mu) \| \le 2.5\varepsilon, \hspace{10pt}
\|(\mu/t)^{-1}(\mu^\new/t^\new - \mu/t) \| \le 3\varepsilon, \hspace{10pt}
\| (u^\new - u) / u \|_2 \le 3 \varepsilon.$$
\fi
\end{lemma}

In order to prove this lemma, we must assume that $\mu$ is currently a good approximation of $t$.
We assume the following proposition, which is proven in the next subsection.

\begin{proposition}\label{pro:mu_eps}
For the input to \Cref{alg:approximate_step} we have
$\mu \approx_{0.1} t$
\end{proposition}

How much we change $x,s$ depends on how long the vector $\delta_\mu$ is, 
so we start by bounding that length.

\begin{lemma}\label{lem:delta_length}
$\| \delta_t \|_2 \le 1.1 \frac{\varepsilon}{3} t$,
$\| \delta_\Phi \|_2 \le \frac{\varepsilon}{2} t$,
$\| \delta_\mu \|_2 \le \varepsilon t$\\
$\| \tdelta_t \|_2 \le 1.2 \frac{\varepsilon}{3} t$,
$\| \tdelta_\Phi \|_2 \le \frac{\varepsilon}{2} t$,
$\| \tdelta_\mu \|_2 \le \varepsilon t$
\end{lemma}

\begin{proof}
\ifdefined\SODAversion
\begin{align*}
\|\delta_t\|_2
&= \|(t^\new / t - 1) \mu \|_2
\le 1.1 \sqrt{n} (1 - t^\new/t) t \\
&= 1.1 \sqrt{n}\frac{\varepsilon}{3\sqrt{n}} t = 1.1 \frac{\varepsilon}{3} t
\end{align*}
\else
$$
\|\delta_t\|_2
= \|(t^\new / t - 1) \mu \|_2
\le 1.1 \sqrt{n} (1 - t^\new/t) t
= 1.1 \sqrt{n}\frac{\varepsilon}{3\sqrt{n}} t = 1.1 \frac{\varepsilon}{3} t
$$
\fi
Here the first inequality comes from $\mu \approx_{0.1} t$. 
This then also implies $\| \tdelta_t \|_2 \le 1.2 \frac{\varepsilon}{3} t$, because $\delta_t \approx_{\varepsilon_{mp}} \tdelta_t$ from \Cref{lem:error_delta}. Next we handle the length of $\delta_\Phi$:
\ifdefined\SODAversion
\begin{align*}
\| \delta_\Phi\|_2
&= \left\| \frac{\varepsilon}{2} t^\new \frac{\nabla \Phi_\lambda(\mu/t-1)}{\|\nabla \Phi_\lambda(\mu/t-1)\|_2}\right\|_2 \\
&= \frac{\varepsilon t^\new}{2} 
= \frac{\varepsilon (1-\frac{\varepsilon}{3 \sqrt{n}}) t}{2} \le \frac{\varepsilon}{2} t
\end{align*}
\else
$$
\| \delta_\Phi\|_2
= \left\| \frac{\varepsilon}{2} t^\new \frac{\nabla \Phi_\lambda(\mu/t-1)}{\|\nabla \Phi_\lambda(\mu/t-1)\|_2}\right\|_2
= \frac{\varepsilon t^\new}{2}
= \frac{\varepsilon (1-\frac{\varepsilon}{3 \sqrt{n}}) t}{2} \le \frac{\varepsilon}{2} t
$$
\fi
The same proof also yields the bound for $\tdelta_\Phi$ as we just replace $\mu/t$ by $\tilde{\mu}/t$, but because of the normalization this does not change the length.
By combining the past results via triangle inequality we obtain
$$
\| \delta_\mu \|
\le
\|\delta_t\|_2 + \| \delta_\Phi\|_2
\le 1.1 \frac{\varepsilon}{3} t + \frac{\varepsilon}{2} t
\le \varepsilon t
$$
and likewise $\| \tdelta_\mu \| \le \varepsilon t$.
\end{proof}

Next we show that the multiplicative change to $x$ and $s$ is small.
\begin{lemma}\label{lem:change_sx}
$\| \tilde{s}^{-1} \tdelta_s \|_2 \le 1.2 \varepsilon$,
$\| s^{-1} \tdelta_s \|_2 \le 1.2 \varepsilon$,\\
$\| \tilde{x}^{-1} \tdelta_x \|_2 \le 1.2 \varepsilon$,
$\| x^{-1} \tdelta_x \|_2 \le 1.2 \varepsilon$%,\\
%$\| \mu^{-1} \tdelta_\mu \|_2 \le 2.4 \varepsilon$,
\end{lemma}

\begin{proof}
Since $\tilde{P}$ is an orthogonal projection matrix we have 
$\|\tilde{P} \frac{\tdelta_\mu}{\sqrt{\tilde{X}\tilde{S}}}\|_2 
\le \|\frac{\tdelta_\mu}{\sqrt{\tilde{X}\tilde{S}}}\|_2$ 
and as $\mu \approx_{\varepsilon_{mp}} \tilde{\mu} = \tilde{x}\tilde{s}$ and $\mu \approx_{0.1} t$,
this can be further bounded by 
$\sqrt{(1+\varepsilon_{mp})/(0.9t)} \|\tdelta_\mu\| 
%\le \sqrt{(1+\varepsilon_{mp})/(0.9t)}\varepsilon t
%\le 1.1 \varepsilon \sqrt{t}
.$
This allows us to bound $\| \tilde{s}^{-1} \tdelta_s \|_2 $ as follows:
\ifdefined\SODAversion
\begin{align*}
&~\| \tilde{s}^{-1} \tdelta_s \|_2 
= 
\| \frac{1}{\sqrt{\tilde{X}\tilde{S}}} \tilde{P} \frac{\tdelta_\mu}{\sqrt{\tilde{X}\tilde{S}}} \|_2 \\
\le&~ 
\sqrt{(1+\varepsilon_{mp})/(0.9t)} \| \tilde{P} \frac{\tdelta_\mu}{\sqrt{\tilde{X}\tilde{S}}} \|_2\\
\le&~
(1+\varepsilon_{mp})/(0.9t) \|\tdelta_\mu\|_2
\le (1+\varepsilon_{mp})/0.9 \varepsilon
\le 1.2 \varepsilon,
\end{align*}
\else
\begin{align*}
\| \tilde{s}^{-1} \tdelta_s \|_2 
&= 
\| \frac{1}{\sqrt{\tilde{X}\tilde{S}}} \tilde{P} \frac{\tdelta_\mu}{\sqrt{\tilde{X}\tilde{S}}} \|_2
\le 
\sqrt{(1+\varepsilon_{mp})/(0.9t)} \| \tilde{P} \frac{\tdelta_\mu}{\sqrt{\tilde{X}\tilde{S}}} \|_2\\
&\le
(1+\varepsilon_{mp})/(0.9t) \|\tdelta_\mu\|_2
\le (1+\varepsilon_{mp})/0.9 \varepsilon
\le 1.2 \varepsilon,
\end{align*}
\fi
where we used $\|\tdelta_\mu\| \le \varepsilon t$ 
from \Cref{lem:delta_length}.
The proof for $\| \tilde{x}^{-1} \tdelta_x \|_2 \le 1.2 \varepsilon$ is identical as $\I - \tilde{P}$ is also a projection matrix.

As $x \approx_{\varepsilon_{mp}} \tilde{x}$,
$s \approx_{2\varepsilon_{mp}} \tilde{s}$ we have
$\| s^{-1} \tdelta_s \|_2 \le (1-\varepsilon_{mp})^{-1} \| \tilde{s}^{-1} \tdelta_s \|_2 \le 1.2 \varepsilon$,
$\| x^{-1} \tdelta_x \|_2 \le (1-\varepsilon_{mp})^{-1}\| \tilde{x}^{-1} \tdelta_x \|_2 \le 1.2 \varepsilon$ via the same proof.

\end{proof}

With this we can now prove \Cref{lem:change_mu_u}. 
We split the proof into two separate corollaries: 
one for $\mu$ and one for $u$.

\ifdefined\SODAversion
\begin{corollary}\label{lem:change_mu}
$\|\mu^{-1}(\mu^\new - \mu) \| \le 2.5\varepsilon$,\\
$\|(\mu/t)^{-1}(\mu^\new/t^\new - \mu/t) \| \le 3\varepsilon$
\end{corollary}
\else
\begin{corollary}\label{lem:change_mu}
$\|\mu^{-1}(\mu^\new - \mu) \| \le 2.5\varepsilon$, 
$\|(\mu/t)^{-1}(\mu^\new/t^\new - \mu/t) \| \le 3\varepsilon$
\end{corollary}
\fi

\begin{proof}
The first claim follows from $\mu = xs$, 
$\mu^\new = (x+\tdelta_x)(s+\tdelta_s)$
and $\|x^{-1}\tdelta_x\|,\|s^{-1}\tdelta_s\| \le 1.2\varepsilon$, 
and applying \Cref{lem:product_change}:
\ifdefined\SODAversion
\begin{align*}
&~\|\mu^{-1}(\mu^\new - \mu) \| \\
\le&~
\|x^{-1}\tdelta_x\|_2 + \|s^{-1}\tdelta_s\|_2 + \|x^{-1}\tdelta_x\|_2\|s^{-1}\tdelta_s\|_2 \\ 
\le&~ 
1.2 \varepsilon + 1.2 \varepsilon + (1.2 \varepsilon)^2 \le 2.5 \varepsilon
\end{align*}
\else
$$
\|\mu^{-1}(\mu^\new - \mu) \| \le \|x^{-1}\tdelta_x\|_2 + \|s^{-1}\tdelta_s\|_2 + \|x^{-1}\tdelta_x\|_2\|s^{-1}\tdelta_s\|_2 \le 1.2 \varepsilon + 1.2 \varepsilon + (1.2 \varepsilon)^2 \le 2.5 \varepsilon
$$
\fi

The second claim is implied by \Cref{lem:product_change} and \Cref{lem:inverse_change}:
\Cref{lem:inverse_change} allows us to describe 
how much $(t^\new)^{-1} \cdot \mathbf{1}_n$ changed 
compared to $t^{-1} \cdot \mathbf{1}_n$:
\ifdefined\SODAversion
\begin{align*}
&~ \|t \cdot \mathbf{1}_n ((\frac{1}{t^\new} - \frac{1}{t}) \cdot \mathbf{1}_n)\|_2 \\
\le&~
\frac{\|t^{-1} \cdot \mathbf{1}_n ((t^\new - t) \cdot \mathbf{1}_n)\|_2}{1- \|t^{-1} \cdot \mathbf{1}_n ((t^\new - t) \cdot \mathbf{1}_n)\|_2} \\
\le&~
\frac{\sqrt{n}|(t^\new - t)/t|}{1- \sqrt{n}|(t^\new - t)/t|}
=
\frac{\sqrt{n}\frac{\varepsilon}{3\sqrt{n}}}{1- \frac{\varepsilon}{3\sqrt{n}}}
\le
0.35 \varepsilon
\end{align*}
\else
$$
\|t \cdot \mathbf{1}_n ((\frac{1}{t^\new} - \frac{1}{t}) \cdot \mathbf{1}_n)\|_2
\le
\frac{\|t^{-1} \cdot \mathbf{1}_n ((t^\new - t) \cdot \mathbf{1}_n)\|_2}{1- \|t^{-1} \cdot \mathbf{1}_n ((t^\new - t) \cdot \mathbf{1}_n)\|_2}
\le 
\frac{\sqrt{n}|(t^\new - t)/t|}{1- \sqrt{n}|(t^\new - t)/t|}
=
\frac{\sqrt{n}\frac{\varepsilon}{3\sqrt{n}}}{1- \frac{\varepsilon}{3\sqrt{n}}}
\le
0.35 \varepsilon
$$
\fi
Then \Cref{lem:inverse_change} tells us
$
\|(\mu/t)^{-1}(\mu^\new/t^\new - \mu/t) \|
\le
0.35 \varepsilon + 2.5\varepsilon + (0.35 \cdot 2.5) \varepsilon^2
\le 3 \varepsilon. 
$
\end{proof}

Likewise, the multiplicative change of $u := \frac{x}{s}$ can be bounded as follows:

\begin{corollary}\label{lem:change_u}
Let $u := \frac{x}{s}$, then
$\| (u^\new - u) / u \|_2 \le 3 \varepsilon$
\end{corollary}

\begin{proof}
We have $\|x^{-1} \delta_x\|_2, \|s^{-1} \delta_s \| \le 1.2 \varepsilon$, 
see \Cref{lem:change_sx}. 
Thus $\|s((s+\delta_s)^{-1}-s^{-1})\| 
\le 1.2\varepsilon/(1-1.2\varepsilon) 
\le 1.4 \varepsilon$ 
by \Cref{lem:inverse_change}.
This leads to 
$\|u^{-1}(u^\new - u) \| 
\le 1.4\varepsilon + 1.2\varepsilon + (1.2\varepsilon)^2 
< 3 \varepsilon$, 
because of $u = x/s$ and \Cref{lem:product_change}.
\end{proof}

\subsection{Maintaining $\mu \approx t$}
\label{sub:mu_close_to_t}

In this section we prove \Cref{pro:mu_eps}, so $\mu \approx_{0.1} t$.
An alternative way to write this statement is $\|\mu/t-1\|_\infty \le 0.1$.
We prove that this norm is small, 
by showing that the potential $\Phi_\lambda(\mu/t-1)$ stays below a certain threshold.
The choice of $\Phi_\lambda(x) = \sum_{i=1}^n \cosh(x_i)$ is motivated by the following lemma:
\begin{lemma}\label{lem:bound_infty_via_phi}
$\|\mu/t-1\|_\infty \le \frac{\ln 2\Phi_\lambda(\mu/t-1)}{\lambda}$
\end{lemma}

\begin{proof}
$\Phi_\lambda(x) 
= \frac{1}{2}\sum_{i=1}^n e^{\lambda x_i} + e^{-\lambda x_i} 
\ge \frac{1}{2} e^{\lambda \|x\|_\infty}$,
so $\|x\|_\infty \le \frac{\ln 2\Phi_\lambda(x)}{\lambda}$.
\end{proof}

This means we must prove $\Phi_\lambda(\mu/t-1) \le 0.5 \cdot e^{0.1 \lambda} = 0.5 n^4$.
We prove this in an inductive way. More accurately, in this section we prove the following lemma.
(Note that $2n \le 0.5 n^4$ for $n > 1$.)
\begin{lemma}\label{lem:inductive_step}
If $\Phi_\lambda(\mu/t-1) \le 2n$, then $\Phi_\lambda(\mu^\new/t^\new-1)\le 2n$.
\end{lemma}

In order to show that \Cref{lem:inductive_step} is true, 
we must first bound the impact of all the approximations.
We start by bounding the error that we incur 
based on the approximation $\mu^\new \approx \mu + \tdelta_\mu$, 
when in reality we have $\mu^\new = (x+\tdelta_x)(s+\tdelta_s) = \mu + \tdelta_\mu + \tdelta_x\tdelta_s$.

\begin{lemma}\label{lem:bound_error}
For $\mu^\new = (x + \tdelta_x)(s + \tdelta_s)$ we have
$
\| \mu^\new - \mu - \tdelta_\mu \|_2 \le 6 t \varepsilon^2.
$
\end{lemma}

\begin{proof}
We can expand the term for $\mu^\new$ as follows:
\ifdefined\SODAversion
\begin{align*}
\mu^\new 
&= (x + \tdelta_x)(s + \tdelta_s)
= xs + x\tdelta_s + s\tdelta_x + \tdelta_x \tdelta_s \\
&= \mu + \underbrace{\tilde{x}\tdelta_s + \tilde{s}\tdelta_x}_{\tdelta_\mu}
+ (x-\tilde{x})\tdelta_s + (s-\tilde{s})\tdelta_x + \tdelta_x \tdelta_s.
\end{align*}
Hence the error (relative to $\mu$) can be bounded as follows:
\begin{align*}
&~\| \mu^{-1} (\mu^\new - \mu - \tdelta_\mu) \|_2 \\
=&~
  \| \mu^{-1} ((x-\tilde{x})\tdelta_s 
  + (s-\tilde{s})\tdelta_x + \tdelta_x \tdelta_s) \|_2 \\
\le&~
  \| \mu^{-1} (x - \tilde{x}) s \cdot s^{-1} \tdelta_s\|_2
+ \| \mu^{-1} (s - \tilde{s}) x \cdot x^{-1} \tdelta_x\|_2 \\
&~+ \| \mu^{-1} \tdelta_x \tdelta_s\|_2 \\
\le&~
  \| \mu^{-1} (x - \tilde{x}) s\|_\infty \| s^{-1} \tdelta_s\|_2 \\
&~+ \| \mu^{-1} (s - \tilde{s}) x\|_\infty \| x^{-1} \tdelta_x\|_2 
  + \| x^{-1} \tdelta_x s^{-1} \tdelta_s\|_2 \\
\le&~
  \frac{\varepsilon_{mp}}{1-\varepsilon_{mp}} \| s^{-1} \tdelta_s\|_2 
  + \frac{2\varepsilon_{mp}}{1-2\varepsilon_{mp}} \| x^{-1} \tdelta_x\|_2 \\
&~+ \| x^{-1} \tdelta_x\|_2 \| s^{-1} \tdelta_s\|_2 \\
\le&~
  3.7\varepsilon_{mp}\varepsilon
+ (1.2\varepsilon)^2
\end{align*}
\else
$$\mu^\new 
= (x + \tdelta_x)(s + \tdelta_s)
= xs + x\tdelta_s + s\tdelta_x + \tdelta_x \tdelta_s
= \mu + \underbrace{\tilde{x}\tdelta_s + \tilde{s}\tdelta_x}_{\tdelta_\mu}
+ (x-\tilde{x})\tdelta_s + (s-\tilde{s})\tdelta_x + \tdelta_x \tdelta_s.$$
Hence the error (relative to $\mu$) can be bounded as follows:
\begin{align*}
\| \mu^{-1} (\mu^\new - \mu - \tdelta_\mu) \|_2
&=
  \| \mu^{-1} ((x-\tilde{x})\tdelta_s + (s-\tilde{s})\tdelta_x + \tdelta_x \tdelta_s) \|_2 \\
&\le
  \| \mu^{-1} (x - \tilde{x}) s \cdot s^{-1} \tdelta_s\|_2
+ \| \mu^{-1} (s - \tilde{s}) x \cdot x^{-1} \tdelta_x\|_2
+ \| \mu^{-1} \tdelta_x \tdelta_s\|_2 \\
&\le
  \| \mu^{-1} (x - \tilde{x}) s\|_\infty \| s^{-1} \tdelta_s\|_2
+ \| \mu^{-1} (s - \tilde{s}) x\|_\infty \| x^{-1} \tdelta_x\|_2
+ \| x^{-1} \tdelta_x s^{-1} \tdelta_s\|_2 \\
&\le
  \frac{\varepsilon_{mp}}{1-\varepsilon_{mp}} \| s^{-1} \tdelta_s\|_2 + \frac{2\varepsilon_{mp}}{1-2\varepsilon_{mp}} \| x^{-1} \tdelta_x\|_2
+ \| x^{-1} \tdelta_x\|_2 \| s^{-1} \tdelta_s\|_2 \\
&\le
  3.7\varepsilon_{mp}\varepsilon
+ (1.2\varepsilon)^2
\end{align*}
\fi
For the fourth line we used 
$\mu = xs$, $x \approx_{\varepsilon_{mp}} \tilde{x}$, $s \approx_{2\varepsilon_{mp}} \tilde{s}$,
which implies (for example)
$\mu^{-1}(x-\tilde{x})s = x^{-1}(x-\tilde{x}) \le x^{-1} \varepsilon_{mp} \tilde{x} \le \frac{\varepsilon_{mp}}{1-\varepsilon_{mp}}$.
The last line uses \Cref{lem:change_sx}.

By exploiting $\mu \approx_{0.1} t$ and $\varepsilon_{mp} \le \varepsilon$, 
we get $\|\mu^\new - \mu - \tdelta_\mu \|_2 \le 6 t \varepsilon^2$.

\end{proof}

Another source of error is that $\tdelta_\Phi$ and $\delta_\Phi$ 
(which depend on $\nabla\Phi_\lambda(\tilde{\mu}/t-1)$ and $\nabla \Phi_\lambda(\mu/t-1)$)
might point in two completely different directions. 
This issue was outlined in the overview \Cref{sec:outline:modified_central_path}, 
where we claimed that for $\Phi_\lambda(\mu/t-1)$ large enough, 
the approximate gradient $\nabla\Phi_\lambda(\tilde{\mu}/t-1)$ does point in the same direction as $\nabla \Phi_\lambda(\mu/t-1)$.
In order to prove this claim, we require some properties of $\Phi_\lambda(\cdot)$.

\begin{lemma}[{\cite{CohenLS19}}]
\label{lem:app:taylor}
Let $\Phi_\lambda(x) = \sum_{i=1}^n \cosh(\lambda x_i)$, then
\begin{enumerate}
\item For any $\| v \|_\infty \le 1 / \lambda$ we have
$$
\Phi_\lambda(r+v) \le \phi_\lambda(r) + \langle \nabla \Phi_\lambda(r), v \rangle + 2 \|v\|_{\nabla^2 \phi_\lambda(r)}.
$$
\item $\| \nabla \phi_\lambda(r)\|_2 \ge \frac{\lambda}{\sqrt{n}} (\Phi_\lambda (r) - n)$
\item $(\sum_{i=1}^n \lambda^2 \Phi_\lambda(r_i)_i^2)^{0.5} \le \lambda \sqrt{n} + \| \nabla \Phi_\lambda(r)\|_2$
\end{enumerate}
\end{lemma}

With these tools we can now analyze the impact 
of approximating $\nabla \Phi_\lambda(\mu/t-1)$ 
via $\nabla \Phi_\lambda(\tilde{\mu}/t-1)$.
The following lemma says that, if the potential
$\| \nabla \Phi_\lambda(\mu/t-1)\|_2$ is larger than
$(2.5/0.9)\lambda^2\varepsilon_{mp} \sqrt{n}$, 
then the approximate gradient does point in the correct direction
(i.e. the inner product with the real gradient is positive).

\ifdefined\SODAversion
\begin{lemma}\label{lem:gradient_direction}
\textsc{(Good direction for large $\Phi_\lambda(\mu/t-1)$)}
\begin{align*}
&~\langle
\nabla \Phi_\lambda(\mu/t-1), 
-\frac{\nabla \Phi_\lambda(\tilde{\mu}/t -1)}{\|\nabla \Phi_\lambda(\tilde{\mu}/t -1)\|_2}
\rangle \\
\le&~
- 0.9 \| \nabla \Phi_\lambda(\mu/t-1)\|_2
+ 2.5\lambda^2\varepsilon_{mp} \sqrt{n}
\end{align*}
\end{lemma}

\begin{proof}
\begin{align*}
&~\langle \nabla \Phi_\lambda(\mu/t-1), -\nabla \Phi_\lambda(\tilde{\mu}/t -1) \rangle\\
=&~
- \langle \nabla \Phi_\lambda(\tilde{\mu}/t-1), \nabla \Phi_\lambda(\tilde{\mu}/t -1) \rangle \\
&~+ \langle \nabla \Phi_\lambda(\mu/t-1) - \nabla \Phi_\lambda(\tilde{\mu}/t-1), \nabla \Phi_\lambda(\tilde{\mu}/t -1) \rangle \\
\le&~
- \| \nabla \Phi_\lambda(\tilde{\mu}/t-1)\|_2^2 \\
&~ + \| \nabla \Phi_\lambda(\mu/t-1) - \nabla \Phi_\lambda(\tilde{\mu}/t-1)\|_2\\
&~ \cdot \| \nabla \Phi_\lambda(\tilde{\mu}/t -1) \|_2
\end{align*}
By normalizing the second vector we then obtain:
\begin{align*}
&~\langle
\nabla \Phi_\lambda(\mu/t-1), 
-\frac{\nabla \Phi_\lambda(\tilde{\mu}/t -1)}{\|\nabla \Phi_\lambda(\tilde{\mu}/t -1)\|_2}
\rangle \\
\le&~
- \| \nabla \Phi_\lambda(\tilde{\mu}/t-1)\|_2 \\
&~+ \| \nabla \Phi_\lambda(\mu/t-1) - \nabla \Phi_\lambda(\tilde{\mu}/t-1)\|_2
\end{align*}
So in order to prove \Cref{{lem:gradient_direction}},
we must bound the norm 
$\| \nabla \Phi_\lambda(\mu/t-1) - \nabla \Phi_\lambda(\tilde{\mu}/t-1)\|_2$.
Note that $\nabla \Phi_\lambda(x)_i = \lambda \sinh(\lambda x_i)$ and $\sinh(x) = (e^x - e^{-x})/2$.
So for now let us bound $|\sinh(x+y)-\sinh(x)|$:
\begin{align*}
&~|\sinh(x+y)-\sinh(x)| \\
=&~
|e^x \cdot e^y - e^{-x} \cdot e^{-y} - (e^x - e^{-x})|/2 \\
=&~
|e^x \cdot (e^y-1) + e^{-x} \cdot (1-e^{-y})|/2 \\
\le&~
(e^x \cdot |e^y-1| + e^{-x} \cdot |1-e^{-y}|)/2 \\
\le&~
(e^x + e^{-x})/2 \cdot \max(|e^y-1|, |1-e^{-y}|) \\
\le&~
(e^x + e^{-x})/2 (e^{|y|}-1)
=
\cosh(x) (e^{|y|}-1)
\end{align*}
Thus we can bound the difference as follows
\begin{align*}
&~\| \nabla \Phi_\lambda(\mu/t-1) - \nabla \Phi_\lambda(\tilde{\mu}/t-1)\|_2 \\
=&~
\lambda \| \sinh(\lambda(\mu/t-1)) - \sinh(\lambda(\tilde{\mu}/t-1))\|_2 \\
=&~
\lambda \| \sinh(\lambda(\tilde{\mu}/t-1+(\mu-\tilde{\mu})/t)) - \sinh(\lambda (\tilde{\mu}/t-1))\|_2 \\
\le&~
\lambda \| \cosh(\lambda(\tilde{\mu}/t-1)) (e^{\lambda|\tilde{\mu}-\mu|/t}-1)\|_2 \\
\le&~
\lambda \| \cosh(\lambda(\tilde{\mu}/t-1))\|_2 (e^{\lambda\|\tilde{\mu}-\mu\|_\infty/t}-1) \\
\le&~
(\lambda \sqrt{n} + \|\nabla \Phi_\lambda(\tilde{\mu}/t-1)\|_2) (e^{\lambda\|\tilde{\mu}-\mu\|_\infty/t}-1)
\end{align*}
For the last inequality we used the third statement of \Cref{lem:app:taylor}.
Note that $\|(\tilde{\mu}-\mu)/t\|_\infty \le \|\varepsilon_{mp}\tilde{\mu}/t\|_\infty \le \|\frac{\varepsilon_{mp}}{1-\varepsilon_{mp}}\mu/t\|_\infty \le \frac{1.1\varepsilon_{mp}}{1-\varepsilon_{mp}}$. 
As $\varepsilon_{mp} \le 1/\lambda$ 
we can use $e^{|x|} \le 1+2|x|$ for $|x| < 1.25$ 
to bound the extra factor 
$(e^{\lambda\|\tilde{\mu}-\mu\|_\infty/t}-1) < 2.5\lambda\varepsilon_{mp}$.

Finally, this allows us to obtain
\begin{align*}
&~\langle
\nabla \Phi_\lambda(\mu/t-1), 
-\frac{\nabla \Phi_\lambda(\tilde{\mu}/t -1)}{\|\nabla \Phi_\lambda(\tilde{\mu}/t -1)\|_2}
\rangle \\
\le&~
- \| \nabla \Phi_\lambda(\tilde{\mu}/t-1)\|_2 \\
&~+ \| \nabla \Phi_\lambda(\mu/t-1) - \nabla \Phi_\lambda(\tilde{\mu}/t-1)\|_2 \\
<&~
- \| \nabla \Phi_\lambda(\tilde{\mu}/t-1)\|_2 \\
&~+ 2.5\lambda\varepsilon_{mp} (\lambda \sqrt{n} + \|\nabla \Phi_\lambda(\tilde{\mu}/t-1)\|_2) \\
\le&~
- 0.9 \| \nabla \Phi_\lambda(\mu/t-1)\|_2
+ 2.5\lambda^2\varepsilon_{mp} \sqrt{n}
\end{align*}
For the last inequality we used $2.5\lambda\varepsilon_{mp} < 0.1$.
\end{proof}
\else
\begin{lemma}[Good direction for large $\Phi_\lambda(\mu/t-1)$]
\label{lem:gradient_direction}
$$
\langle
\nabla \Phi_\lambda(\mu/t-1), 
-\frac{\nabla \Phi_\lambda(\tilde{\mu}/t -1)}{\|\nabla \Phi_\lambda(\tilde{\mu}/t -1)\|_2}
\rangle
\le
- 0.9 \| \nabla \Phi_\lambda(\mu/t-1)\|_2
+ 2.5\lambda^2\varepsilon_{mp} \sqrt{n}
$$
\end{lemma}

\begin{proof}
\begin{align*}
&\langle \nabla \Phi_\lambda(\mu/t-1), -\nabla \Phi_\lambda(\tilde{\mu}/t -1) \rangle\\
=&\:
- \langle \nabla \Phi_\lambda(\tilde{\mu}/t-1), \nabla \Phi_\lambda(\tilde{\mu}/t -1) \rangle
+ \langle \nabla \Phi_\lambda(\mu/t-1) - \nabla \Phi_\lambda(\tilde{\mu}/t-1), \nabla \Phi_\lambda(\tilde{\mu}/t -1) \rangle \\
\le&\:
- \| \nabla \Phi_\lambda(\tilde{\mu}/t-1)\|_2^2
+ \| \nabla \Phi_\lambda(\mu/t-1) - \nabla \Phi_\lambda(\tilde{\mu}/t-1)\|_2 \cdot \| \nabla \Phi_\lambda(\tilde{\mu}/t -1) \|_2
\end{align*}
By normalizing the second vector we then obtain:
$$
\langle
\nabla \Phi_\lambda(\mu/t-1), 
-\frac{\nabla \Phi_\lambda(\tilde{\mu}/t -1)}{\|\nabla \Phi_\lambda(\tilde{\mu}/t -1)\|_2}
\rangle
\le
- \| \nabla \Phi_\lambda(\tilde{\mu}/t-1)\|_2
+ \| \nabla \Phi_\lambda(\mu/t-1) - \nabla \Phi_\lambda(\tilde{\mu}/t-1)\|_2
$$

So in order to prove \Cref{{lem:gradient_direction}},
we must bound the norm 
$\| \nabla \Phi_\lambda(\mu/t-1) - \nabla \Phi_\lambda(\tilde{\mu}/t-1)\|_2$.
Note that $\nabla \Phi_\lambda(x)_i = \lambda \sinh(\lambda x_i)$ and $\sinh(x) = (e^x - e^{-x})/2$.
So for now let us bound $|\sinh(x+y)-\sinh(x)|$:
\begin{align*}
|\sinh(x+y)-\sinh(x)|
&=
|e^x \cdot e^y - e^{-x} \cdot e^{-y} - (e^x - e^{-x})|/2
=
|e^x \cdot (e^y-1) + e^{-x} \cdot (1-e^{-y})|/2 \\
&\le
(e^x \cdot |e^y-1| + e^{-x} \cdot |1-e^{-y}|)/2
\le
(e^x + e^{-x})/2 \cdot \max(|e^y-1|, |1-e^{-y}|) \\
&\le
(e^x + e^{-x})/2 (e^{|y|}-1)
=
\cosh(x) (e^{|y|}-1)
\end{align*}
Thus we can bound the difference as follows
\begin{align*}
\| \nabla \Phi_\lambda(\mu/t-1) - \nabla \Phi_\lambda(\tilde{\mu}/t-1)\|_2
&=
\lambda \| \sinh(\lambda(\mu/t-1)) - \sinh(\lambda(\tilde{\mu}/t-1))\|_2 \\
&=
\lambda \| \sinh(\lambda(\tilde{\mu}/t-1+(\mu-\tilde{\mu})/t)) - \sinh(\lambda (\tilde{\mu}/t-1))\|_2 \\
&\le
\lambda \| \cosh(\lambda(\tilde{\mu}/t-1)) (e^{\lambda|\tilde{\mu}-\mu|/t}-1)\|_2 \\
&\le
\lambda \| \cosh(\lambda(\tilde{\mu}/t-1))\|_2 (e^{\lambda\|\tilde{\mu}-\mu\|_\infty/t}-1) \\
&\le
(\lambda \sqrt{n} + \|\nabla \Phi_\lambda(\tilde{\mu}/t-1)\|_2) (e^{\lambda\|\tilde{\mu}-\mu\|_\infty/t}-1)
\end{align*}
For the last inequality we used the third statement of \Cref{lem:app:taylor}.
Note that $\|(\tilde{\mu}-\mu)/t\|_\infty \le \|\varepsilon_{mp}\tilde{\mu}/t\|_\infty \le \|\frac{\varepsilon_{mp}}{1-\varepsilon_{mp}}\mu/t\|_\infty \le \frac{1.1\varepsilon_{mp}}{1-\varepsilon_{mp}}$. 
As $\varepsilon_{mp} \le 1/\lambda$ 
we can use $e^{|x|} \le 1+2|x|$ for $|x| < 1.25$ 
to bound the extra factor 
$(e^{\lambda\|\tilde{\mu}-\mu\|_\infty/t}-1) < 2.5\lambda\varepsilon_{mp}$.

Finally, this allows us to obtain
\begin{align*}
\langle
\nabla \Phi_\lambda(\mu/t-1), 
-\frac{\nabla \Phi_\lambda(\tilde{\mu}/t -1)}{\|\nabla \Phi_\lambda(\tilde{\mu}/t -1)\|_2}
\rangle
&\le
- \| \nabla \Phi_\lambda(\tilde{\mu}/t-1)\|_2
+ \| \nabla \Phi_\lambda(\mu/t-1) - \nabla \Phi_\lambda(\tilde{\mu}/t-1)\|_2 \\
&<
- \| \nabla \Phi_\lambda(\tilde{\mu}/t-1)\|_2
+ 2.5\lambda\varepsilon_{mp} (\lambda \sqrt{n} + \|\nabla \Phi_\lambda(\tilde{\mu}/t-1)\|_2) \\
&\le
- 0.9 \| \nabla \Phi_\lambda(\mu/t-1)\|_2
+ 2.5\lambda^2\varepsilon_{mp} \sqrt{n}
\end{align*}
For the last inequality we used $2.5\lambda\varepsilon_{mp} < 0.1$.
\end{proof}
\fi
We now have all tools available to bound $\Phi_\lambda(\frac{\mu^\new}{t^\new}-1)$:

\begin{lemma}\label{lem:decrease_Phi}
$$
\Phi_\lambda(\frac{\mu^\new}{t^\new}-1) 
\le
  \Phi_\lambda(\mu/t-1)
- \frac{\varepsilon}{3} \frac{\lambda}{\sqrt{n}} (\Phi_\lambda(\mu/t-1) - 0.5n)
$$
\end{lemma}

\begin{proof}

First let us write $\frac{\mu^\new}{t^\new} - 1$ 
as $\frac{\mu}{t} - 1 + v$ for some vector $v$. Then
\ifdefined\SODAversion
\begin{align*}
v 
=&~ 
\frac{\mu^\new}{t^\new} - \frac{\mu}{t} \\
=&~
\frac{\mu^\new - \mu - \delta_t - \tdelta_\Phi}{t^\new} + \frac{\mu + \delta_t + \tdelta_\Phi}{t^\new} - \frac{\mu}{t}\\
=&~
\frac{\mu^\new - \mu - \delta_t - \tdelta_\Phi}{t^\new} \\
&~ + \frac{\mu + (t^\new/t-1)\mu - \frac{\varepsilon}{2} t^\new \frac{\nabla \Phi_\lambda(\tilde{\mu}/t-1)}{ \| \Phi_\lambda(\tilde{\mu}/t-1)\|_2}}{t^\new} - \frac{\mu}{t}\\
=&~
\frac{\mu^\new - \mu - \delta_t - \tdelta_\Phi}{t^\new} + \frac{\mu}{t^\new} + \frac{\mu}{t} - \frac{\mu}{t^\new} \\
&~ - \frac{\varepsilon}{2} \frac{\nabla \Phi_\lambda(\tilde{\mu}/t-1)}{ \| \Phi_\lambda(\tilde{\mu}/t-1)\|_2} - \frac{\mu}{t} \\
=&~
\frac{\mu^\new - \mu - \delta_t - \tdelta_\Phi}{t^\new} - \frac{\varepsilon}{2} \frac{\nabla \Phi_\lambda(\tilde{\mu}/t-1)}{ \| \Phi_\lambda(\tilde{\mu}/t-1)\|_2}
\end{align*}
In order to use \Cref{lem:app:taylor}, we must show that $\|v \|_2 < 1/\lambda$. 
For that we bound the length of $\|\frac{\mu^\new - \mu - \delta_t - \tdelta_\Phi}{t^\new}\|$ as follows:
\begin{align*}
&~\|\frac{\mu^\new - \mu - \delta_t - \tdelta_\Phi}{t^\new}\|_2 \\
=&~
\|\frac{\mu^\new - \mu - \tdelta_\mu + (\tdelta_t - \delta_t)}{t^\new}\|_2\\
\le&~
\frac{1}{t^\new} (\| \mu^\new - \mu - \tdelta_\mu \|_2
+ \|\tdelta_t - \delta_t\|) \\
\le&~
\frac{1}{t^\new} (6t\varepsilon^2
+ \varepsilon_{mp} \| \tdelta_t \|_2 ) %\\
\le%&~
\frac{t}{t^\new} (6+\frac{1.2}{3}) \varepsilon^2 < 6.5 \varepsilon^2
\end{align*}
In the first line we used $\tdelta_\mu = \tdelta_t + \tdelta_\Phi$ 
and in the last line we used \Cref{lem:bound_error,lem:delta_length}.
Thus $\|v\|_2 \le 6.5 \varepsilon^2 + \varepsilon/2 < \varepsilon \le 1/\lambda$ and we can apply \Cref{lem:app:taylor}:
\begin{align*}
&~\Phi_\lambda(\mu/t+v-1) \\
\le&~
\Phi_\lambda(\mu/t-1) + \langle \nabla \Phi_\lambda(\mu/t-1), v \rangle + 2 \|v\|^2_{\nabla^2 \Phi_\lambda(\mu/t-1)} \\
=&~
  \Phi_\lambda(\mu/t-1) 
- \frac{\varepsilon}{2}\langle \Phi_\lambda(\mu/t-1), \frac{\nabla \Phi_\lambda(\tilde{\mu}/t-1)}{ \| \Phi_\lambda(\tilde{\mu}/t-1)\|_2} \rangle \\
&~ + \langle \nabla \Phi_\lambda(\mu/t-1), \frac{\mu^\new - \mu - \delta_t - \tdelta_\Phi}{t^\new} \rangle \\
&~ + 2 \|v\|^2_{\nabla^2 \Phi_\lambda(\mu/t-1)}\\
\le&~
  \Phi_\lambda(\mu/t-1) 
- \frac{0.9\varepsilon}{2}\|\nabla \Phi_\lambda(\mu/t-1)\|_2 + 1.25 \varepsilon^2 \lambda^2 \sqrt{n}\\
&~ + \|\nabla \Phi_\lambda(\mu/t-1)\|_2 \cdot \|\frac{\mu^\new - \mu - \delta_t - \tdelta_\Phi}{t^\new}\|_2 \\
&~ + 2 \|v\|^2_{\nabla^2 \Phi_\lambda(\mu/t-1)} \\
\le&~
  \Phi_\lambda(\mu/t-1) 
- \frac{0.9\varepsilon}{2}\|\nabla \Phi_\lambda(\mu/t-1)\|_2 \\
&~ + 1.25 \varepsilon^2 \lambda^2 \sqrt{n}
   + 6.5 \varepsilon^2 \|\nabla \Phi_\lambda(\mu/t-1)\|_2 \\
&~ + 2 \|v\|^2_{\nabla^2 \Phi_\lambda(\mu/t-1)}
\end{align*}
\else
\begin{align*}
v 
&= 
\frac{\mu^\new}{t^\new} - \frac{\mu}{t}
=
\frac{\mu^\new - \mu - \delta_t - \tdelta_\Phi}{t^\new} + \frac{\mu + \delta_t + \tdelta_\Phi}{t^\new} - \frac{\mu}{t}\\
&=
\frac{\mu^\new - \mu - \delta_t - \tdelta_\Phi}{t^\new} + \frac{\mu + (t^\new/t-1)\mu - \frac{\varepsilon}{2} t^\new \frac{\nabla \Phi_\lambda(\tilde{\mu}/t-1)}{ \| \Phi_\lambda(\tilde{\mu}/t-1)\|_2}}{t^\new} - \frac{\mu}{t}\\
&=
\frac{\mu^\new - \mu - \delta_t - \tdelta_\Phi}{t^\new} + \frac{\mu}{t^\new} + \frac{\mu}{t} - \frac{\mu}{t^\new} - \frac{\varepsilon}{2} \frac{\nabla \Phi_\lambda(\tilde{\mu}/t-1)}{ \| \Phi_\lambda(\tilde{\mu}/t-1)\|_2} - \frac{\mu}{t} \\
&=
\frac{\mu^\new - \mu - \delta_t - \tdelta_\Phi}{t^\new} - \frac{\varepsilon}{2} \frac{\nabla \Phi_\lambda(\tilde{\mu}/t-1)}{ \| \Phi_\lambda(\tilde{\mu}/t-1)\|_2}
\end{align*}
In order to use \Cref{lem:app:taylor}, we must show that $\|v \|_2 < 1/\lambda$. 
For that we bound the length of $\|\frac{\mu^\new - \mu - \delta_t - \tdelta_\Phi}{t^\new}\|$ as follows:
\begin{align*}
\|\frac{\mu^\new - \mu - \delta_t - \tdelta_\Phi}{t^\new}\|_2
&=
\|\frac{\mu^\new - \mu - \tdelta_\mu + (\tdelta_t - \delta_t)}{t^\new}\|_2\\
&\le
\frac{1}{t^\new} (\| \mu^\new - \mu - \tdelta_\mu \|_2
+ \|\tdelta_t - \delta_t\|) \\
&\le
\frac{1}{t^\new} (6t\varepsilon^2
+ \varepsilon_{mp} \| \tdelta_t \|_2 )
\le
\frac{t}{t^\new} (6+\frac{1.2}{3}) \varepsilon^2 < 6.5 \varepsilon^2
\end{align*}
In the first line we used $\tdelta_\mu = \tdelta_t + \tdelta_\Phi$ 
and in the last line we used \Cref{lem:bound_error,lem:delta_length}.
Thus $\|v\|_2 \le 6.5 \varepsilon^2 + \varepsilon/2 < \varepsilon \le 1/\lambda$ and we can apply \Cref{lem:app:taylor}:
\begin{align*}
\Phi_\lambda(\mu/t+v-1)
&\le
\Phi_\lambda(\mu/t-1) + \langle \nabla \Phi_\lambda(\mu/t-1), v \rangle + 2 \|v\|^2_{\nabla^2 \Phi_\lambda(\mu/t-1)} \\
&=
  \Phi_\lambda(\mu/t-1) 
- \frac{\varepsilon}{2}\langle \Phi_\lambda(\mu/t-1), \frac{\nabla \Phi_\lambda(\tilde{\mu}/t-1)}{ \| \Phi_\lambda(\tilde{\mu}/t-1)\|_2} \rangle \\
&\hspace{50pt} + \langle \nabla \Phi_\lambda(\mu/t-1), \frac{\mu^\new - \mu - \delta_t - \tdelta_\Phi}{t^\new} \rangle
+ 2 \|v\|^2_{\nabla^2 \Phi_\lambda(\mu/t-1)}\\
&\le
  \Phi_\lambda(\mu/t-1) 
- \frac{0.9\varepsilon}{2}\|\nabla \Phi_\lambda(\mu/t-1)\|_2 + 1.25 \varepsilon^2 \lambda^2 \sqrt{n}\\
&\hspace{50pt} + \|\nabla \Phi_\lambda(\mu/t-1)\|_2 \cdot \|\frac{\mu^\new - \mu - \delta_t - \tdelta_\Phi}{t^\new}\|_2
 + 2 \|v\|^2_{\nabla^2 \Phi_\lambda(\mu/t-1)}\\
&\le
  \Phi_\lambda(\mu/t-1) 
- \frac{0.9\varepsilon}{2}\|\nabla \Phi_\lambda(\mu/t-1)\|_2 + 1.25 \varepsilon^2 \lambda^2 \sqrt{n}\\
&\hspace{50pt} + 6.5 \varepsilon^2 \|\nabla \Phi_\lambda(\mu/t-1)\|_2 
+ 2 \|v\|^2_{\nabla^2 \Phi_\lambda(\mu/t-1)}
\end{align*}
\fi
In the third line we used \Cref{lem:gradient_direction} and Cauchy-Schwarz and the last line comes from the bound we proved above.
Next we bound the second order term:
\begin{align*}
\|v\|^2_{\nabla^2 \Phi_\lambda(\mu/t-1)}
&=
\lambda \sum_{i=1}^n \lambda \Phi_\lambda(\mu/t-1)_i v_i^2 \\
&\le
\lambda (\sum_{i=1}^n (\lambda \Phi_\lambda(\mu/t-1)_i)^2)^{0.5} (\sum_{i=1}^n v_i^4)^{0.5} \\
&\le
\lambda (\lambda \sqrt{n} + \| \nabla \Phi_\lambda(\mu/t-1)\|_2) \|v\|_4^2 \\
&\le
\lambda (\lambda \sqrt{n} + \| \nabla \Phi_\lambda(\mu/t-1)\|_2) (\varepsilon)^2
\end{align*}
The first inequality comes form Cauchy-Schwarz, 
the second inequality from \Cref{lem:app:taylor} 
and the last inequality uses $\|v\|_4 \le \|v\|_2 < \varepsilon$.
Plugging all these bound together we obtain:
\ifdefined\SODAversion
\begin{align*}
&~\Phi_\lambda(\mu/t+v-1) \\
\le&~
  \Phi_\lambda(\mu/t-1) 
- \frac{0.9\varepsilon}{2}\|\nabla \Phi_\lambda(\mu/t-1)\|_2
+ 1.25 \varepsilon^2 \lambda^2 \sqrt{n}\\
&~ + 6.5 \varepsilon^2 \| \nabla \Phi_\lambda(\mu/t-1) \|_2 
+ 2\lambda(\lambda \sqrt{n} \\
&~ + \| \nabla \Phi_\lambda(\mu/t-1)\|_2) \varepsilon^2 \\
<&~
  \Phi_\lambda(\mu/t-1) \\
&~ + \|\nabla \Phi_\lambda(\mu/t-1)\|_2
( 6.5 \varepsilon^2 + 2\lambda \varepsilon^2 - \frac{0.9\varepsilon}{2}) \\
&~ + 3.25 \varepsilon^2 \lambda^2 \sqrt{n} \\
<&~
  \Phi_\lambda(\mu/t-1)
- \frac{\varepsilon}{3} \|\nabla \Phi_\lambda(\mu/t-1)\|_2
+ 3.25 \varepsilon^2 \lambda^2 \sqrt{n} \\
\le\:&
  \Phi_\lambda(\mu/t-1)
- \frac{\varepsilon}{3} \frac{\lambda}{\sqrt{n}} (\Phi_\lambda(\mu/t-1) - n)
+ 3.25 \varepsilon^2 \lambda^2 \sqrt{n} \\
\le\:&
  \Phi_\lambda(\mu/t-1)
- \frac{\varepsilon}{3} \frac{\lambda}{\sqrt{n}} (\Phi_\lambda(\mu/t-1) - 10n\varepsilon\lambda) \\
\le\:&
  \Phi_\lambda(\mu/t-1)
- \frac{\varepsilon}{3} \frac{\lambda}{\sqrt{n}} (\Phi_\lambda(\mu/t-1) - 0.5n)
\end{align*}
\else
\begin{align*}
&\Phi_\lambda(\mu/t+v-1) \\
\le\:&
  \Phi_\lambda(\mu/t-1) 
- \frac{0.9\varepsilon}{2}\|\nabla \Phi_\lambda(\mu/t-1)\|_2
+ 1.25 \varepsilon^2 \lambda^2 \sqrt{n}\\
&\hphantom{\Phi_\lambda(\mu/t-1)} + 6.5 \varepsilon^2 \| \nabla \Phi_\lambda(\mu/t-1) \|_2 
+ 2\lambda(\lambda \sqrt{n} + \| \nabla \Phi_\lambda(\mu/t-1)\|_2) \varepsilon^2 \\
<\:&
  \Phi_\lambda(\mu/t-1) 
+ \|\nabla \Phi_\lambda(\mu/t-1)\|_2
( 6.5 \varepsilon^2 + 2\lambda \varepsilon^2 - \frac{0.9\varepsilon}{2})
+ 3.25 \varepsilon^2 \lambda^2 \sqrt{n} \\
<\:&
  \Phi_\lambda(\mu/t-1)
- \frac{\varepsilon}{3} \|\nabla \Phi_\lambda(\mu/t-1)\|_2
+ 3.25 \varepsilon^2 \lambda^2 \sqrt{n} \\
\le\:&
  \Phi_\lambda(\mu/t-1)
- \frac{\varepsilon}{3} \frac{\lambda}{\sqrt{n}} (\Phi_\lambda(\mu/t-1) - n)
+ 3.25 \varepsilon^2 \lambda^2 \sqrt{n} \\
\le\:&
  \Phi_\lambda(\mu/t-1)
- \frac{\varepsilon}{3} \frac{\lambda}{\sqrt{n}} (\Phi_\lambda(\mu/t-1) - 10n\varepsilon\lambda) \\
\le\:&
  \Phi_\lambda(\mu/t-1)
- \frac{\varepsilon}{3} \frac{\lambda}{\sqrt{n}} (\Phi_\lambda(\mu/t-1) - 0.5n)
\end{align*}
\fi
Here the first inequality uses $\|v\|^2_{\nabla^2 \Phi_\lambda(\mu/t-1)} \le \lambda\varepsilon^2 (\lambda \sqrt{n} + \| \nabla \Phi_\lambda(\mu/t-1)\|_2)$.
The third uses $\varepsilon \le 1/(1500 \ln n)$ and $\lambda = 40 \ln n$, so $( 6.5 \varepsilon^2 + 2\lambda \varepsilon^2 - \frac{0.9\varepsilon}{2}) < (6.5/1500 + 2\cdot40/1500 - 0.9/2)\varepsilon < -\varepsilon / 3$.
The fourth inequality uses part 2 of \Cref{lem:app:taylor}.

\end{proof}

\begin{proof}[Proof of \Cref{lem:inductive_step}]
On one hand, \Cref{lem:decrease_Phi} implies that 
$\Phi_\lambda(\mu^\new/t^\new-1) < \Phi_\lambda(\mu/t-1)$, 
if $\Phi_\lambda(\mu/t-1) > 0.5n$.
On the other hand, if $\Phi_\lambda(\mu/t-1) \le 0.5n$,
then $\Phi_\lambda(\mu^\new/t^\new-1) 
\Phi_\lambda(\mu/t-1)
+ \frac{\varepsilon}{3} \frac{\lambda}{\sqrt{n}} 0.5n
\le
\Phi_\lambda(\mu/t-1) + 0.005\sqrt{n}
< 2n$.
Thus if $\Phi_\lambda(\mu/t-1) \le 2n$, then $\Phi_\lambda(\mu^\new/t^\new-1) \le 2n$.

\end{proof}

We now have all intermediate results required to prove our main result of \Cref{thm:LP}.

\begin{proof}[Proof of \Cref{thm:LP}]
We start by proving the correctness:
\paragraph{Correctness of the algorithm}
At the start of algorithm we transform the linear program
as specified in \Cref{lem:feasible_LP} 
to obtain a feasible solution $(x,y,s)$.
For that transformation we choose 
$\gamma = \min\{\delta, 1/\lambda\}$, 
so $\mu - 1 = \gamma c/L$ and 
$ \| \mu/t-1\|_\infty \le 1/\lambda$ 
for $t = 1$ at the start of the algorithm.
This then implies $\Phi_\lambda(\mu/t-1) \le n \cosh(\lambda/\lambda) \le n(1+e)/2 \le 2n$
which for $n > 1$ is less than $0.5 n^4$, 
and thus $\| \mu/t-1\|_\infty \le 0.1$ 
throughout the entire algorithm 
by \Cref{lem:inductive_step,lem:bound_infty_via_phi}.
(This then also proves \Cref{pro:mu_eps}.)

The algorithm runs until $t < \delta^2/(2n)$,
then we have 
$\| \mu \|_1 \le n \| \mu \|_\infty \le 1.1 n t \le \delta^2 \le \gamma^2$,
so we obtain a solution via \Cref{lem:feasible_LP}.

\paragraph{Complexity of the algorithm}

In each iteration, $t$ decreases by a factor of $(1-\frac{\varepsilon}{3\sqrt{n}})$, 
so it takes $O(\sqrt{n}\varepsilon^{-1} \log(\delta/n))$ iterations 
to reach $t < \delta^2/(2n)$.
We now bound the cost per iteration.
The vectors
$u := x/s$,  
$\mu := xs$, 
and $\mu/t$ 
of \Cref{alg:approximate_step} have small multiplicative change, 
bounded by $3\varepsilon$, $2.5\varepsilon$, and $3\varepsilon$ respectively (\Cref{lem:change_u,lem:change_mu}).
Thus the amortized cost per iteration is 
$O(\varepsilon/\varepsilon_{mp}(n^{\omega - 1/2}+n^{2 - a/2 +o(1)}) \log n + n^{1+a})$ via \Cref{lem:projection_maintenance}.
For $\varepsilon_{mp} = \varepsilon = 1/(1500 \ln n)$ and $a = \min\{\alpha,2/3\}$ 
this is $O(n^{\omega - 1/2} \log n)$ for current 
$\omega \approx 2.37$, $\alpha\approx 0.31$
\cite{Williams12,Gall14,GallU18}.
\ifdefined\SODAversion
The total cost is 
$$O((n^\omega + n^{2.5-\alpha/2+o(1)} + n^{2+1/6+o(1)}) \log^2 (n) \log (n/\delta))$$
\else

The total cost is $O((n^\omega + n^{2.5-\alpha/2+o(1)} + n^{2+1/6+o(1)}) \log^2 (n) \log (n/\delta))$ 
\fi
and for current $\omega$, $\alpha$ this is just $O(n^\omega \log^2(n) \log(n/\delta))$.
\end{proof}

\section{Open Problems}

The $\tilde{O}(n^\omega)$ upper bound presented in this paper 
(but also the one from \cite{CohenLS19}) 
seems optimal in the sense, that all known linear system solvers 
require up to $O(n^{\omega})$ time for solving $Ax = b$.
However, this claimed optimality has two caveats:
(i) The algorithm is only optimal when assuming $d = \Omega(n)$. What improvements are possible for $d \ll n$?
(ii) The $\tilde{O}(n^\omega)$ upper bound only holds for the current bounds of $\omega$ and the dual exponent $\alpha$.
No matter how much $\omega$ and $\alpha$ improve, the presented linear program solver can never beat $\tilde{O}(n^{2+1/6})$ time.
So if in the future some upper bound $\omega < 2+1/6$ is discovered, 
then these linear program solvers are no longer optimal.
One open question is thus, if the algorithm can be improved 
to run in truly $\tilde{O}(n^\omega)$ for every bound on $\omega$, 
or alternatively to prove that $\omega > 2+1/6$.\footnote{%
Recent developments indicate that at least the current techniques 
for fast matrix multiplication do not allow for $\omega < 2+1/6 < 2.168$ \cite{ChristandlVZ19,Alman19,AlmanW18a,AlmanW18b}. 
Another recent work that came out after this paper also rules out $\alpha\ge 0.625$ \cite{ChristandlGLZ20}.}

Another interesting question is, 
if the techniques of this paper can also be applied to other interior point algorithms. 
For example, can they be used to speed-up solvers for semidefinite programming?

\section*{Acknowledgement}
I thank Danupon Nanongkai and Thatchaphol Saranurak for discussions. 
I also thank So-Hyeon Park (Sophie) for her questions and feedback regarding the algorithm.
This project has received funding from the European Research Council (ERC) under the European
Unions Horizon 2020 research and innovation programme under grant agreement No 715672.
The algorithmic descriptions in this paper use latex-code of \cite{CohenLS19}, 
available under CC-BY-4.0\footnote{\url{https://creativecommons.org/licenses/by/4.0/}}

\appendix

\section{Appendix}
\label{app:LP}

\begin{lemma}\label{lem:error_after_k_rounds}
Let $(x^k)_{k\ge 1}$ be a sequence of vectors, 
such that for every $k$ we have $\| (x^{k+1} - x^k)/X^{k} \|_2 \le C < \frac{1}{2}$, where $X^k = diag(x^k)$.
Then there exist at most $O((Ck/\varepsilon)^2)$ many $i$ s.t.
$x^k_i > (1+\varepsilon) x^1$ or $x^k_i < (1-\varepsilon) x^1$.
\end{lemma}
\ifdefined\SODAversion
The proof of \Cref{lem:error_after_k_rounds}
can be found in the full version.
\else
\begin{proof}
For $c \le 0.5$ we have $|\log \frac{x^k_i}{x^1_i}| \le 2|\frac{x^k_i}{x^1_i}-1|$ which allows us to bound the following norm:
\begin{align*}
\| \log \frac{x^k}{x^1} \|_2 
= 
\| \log \prod_{i=1}^{k-1} \frac{x^{i+1}}{x^{i}} \|_2
\le
\sum_{i=1}^{k-1} \| \log \frac{x^{i+1}}{x^{i}} \|_2
\le
2 \sum_{i=1}^{k-1} \| \frac{x^{i+1}}{x^{i}} - 1 \|_2
\le
2 k C
\end{align*}
Let $T$ be the number of indices $i$ with $x^k_i \ge (1+\varepsilon) x^1_i$ or $x^k_i \le (1-\varepsilon) x^1_i$. We want to find an upper bound of $T$.

Without loss of generality we can also assume that $x^k$ and $x^1$ differ in at most $T+1$ entries. 
The reason is as follows:
Let's say we are allowed to choose the sequence of $x^1,...,x^k$ and we want to maximize $T$.
Assume there is more than one index $i$ with $x^k_i \ge (1+\varepsilon) x^1_i$ or $x^k_i \le (1-\varepsilon) x^1_i$.
Let $i\neq j$ be two such indices, then we could have tried to increase $T$ by not changing the $j$th entry and changing $i$th entry a bit more.

This leads to
$T \cdot \log (1+\varepsilon) 
\le \| \log \frac{x^k}{x^1} \|_1 
\le \sqrt{T+1} \| \log \frac{x^k}{x^1} \|_2 
\le 2\sqrt{T}kC$ 
which can be reordered to
$T = O((k C/\varepsilon)^2)$.

\end{proof}
\fi

\ifdefined\SODAversion
\else
\begin{lemma}\label{lem:appendix:cohens_bound}
The preprocessing requires $O(n^2d^{\omega-2})$ time.
After $T$ updates the total time of all updates of \Cref{alg:projection_maintenance}, 
when ignoring the branch for $k < n^a$ (so we assume that branch of \cref{line:else_branch} has cost $0$), 
is
$$
O(T \cdot C / \varepsilon_{mp} (n^{\omega-1/2} + n^{2-a/2+o(1)}) \log n).
$$
\end{lemma}

\begin{proof}
The preprocessing cost is dominated by computing $M = A^\top(AUA^\top)^{-1}A$, which takes $O(n^2d^{\omega-2})$ time.
For the update complexity, we first modify the algorithm a bit.
We replace the loop of \cref{line:extra_entries} by:
$k \leftarrow 2^\ell$ for the smallest integer $\ell$ with $y_{\pi(2^\ell)} < (1-0.5 \ell/\log n) \varepsilon_{mp}$. ($\ell = \log n$ if no such $\ell$ exists.)

As we ignore the branch for $k < n^a$ in the complexity analysis, 
we are left with analyzing the cost of performing a rank $k = 2^\ell$ update via the Sherman-Morrison-Woodbury identity.
The cost for this is $O(n^{\omega(1,1,\ell/\log n)})$, 
where $\omega(a,b,c)$ refers to the number of arithmetic operations required 
to compute the matrix product of an $n^a \times n^b$ with an $n^b \times n^c$ matrix.
So the total cost for $T$ calls to \textsc{Update} is bounded by
$$
\sum_{\ell=0}^{\log n} (\text{number of rank $2^\ell$ updates}) \cdot O(n^{\omega(1,\ell/\log n,1)}).
$$
We now prove that the number of rank $2^\ell$ updates is at most $O(T (C/\varepsilon_{mp}) 2^{-\ell/2} \log n)$ 
by showing that there must be at least $\Omega((\varepsilon_{mp}/C)2^{\ell/2} \log^{-1} n)$ calls to \textsc{Update} between any two rank $2^\ell$ updates.

After a rank $2^\ell$ update, we have by choice of $\ell$ that $|u^\new_i/\tilde{u}_i-1| < (1-0.5 \ell/\log n) \varepsilon_{mp}$ for all $i$.
Let $u^{(0)}$ be the input vector $u^\new$ to \textsc{Update}, when we performed the rank $2^\ell$ update 
and let $u^{(1)},u^{(2)},...,$ be the input sequence to all further calls to \textsc{Update} from that point on.
Likewise let $\tilde{u}^{(0)},\tilde{u}^{(1)},...$ be the internal vectors of the data-structure after these calls to \textsc{Update}.
Then we have 
$$|u^{(0)}_i/\tilde{u}^{(0)}_i-1| < (1-0.5 \ell/\log n) \varepsilon_{mp}$$
for all $i$,
but when we perform another rank $2^\ell$ update some $t$ calls to \textsc{Update} later, 
we have at least $2^{\ell-1}$ indices $i$ with 
$$|u^{(t)}_i/\tilde{u}^{(t-1)}_i-1| \ge (1-0.5 (\ell-1)/\log n) \varepsilon_{mp}.$$
That means either $u^{(t)}_i$ differs to $u^{(0)}_i$ by some $(1\pm \Omega(\varepsilon_{mp}/\log n))$-factor,
or $\tilde{u}^{(t-1)}_i$ differs to $u^{(0)}_i$ by some $(1\pm \Omega(\varepsilon_{mp}/\log n))$-factor
(which means there exists some $t'<t$ where $u^{(t')}_i$ differs to $u^{(0)}_i$ by some $(1\pm \Omega(\varepsilon_{mp}/\log n))$-factor,
which caused $\tilde{u}_i$ to receive an update).

So in summary, we know there must be at least $2^{\ell-1}$ indices $i$
for which the input vectors $u$ changed by some $(1\pm \Omega(\varepsilon_{mp}/\log n))$-factor compared to $u^{(0)}$.
By \Cref{lem:error_after_k_rounds} this can happen only after at least $\Omega(\varepsilon_{mp} C^{-1} 2^{\ell/2} \log^{-1} n)$ calls to \textsc{Update}, 
as the multiplicative change between any $u^{(k)}$ and $u^{(k+1)}$ is bounded by $C$.

Note that by definition we only perform rank $2^\ell \ge n^a$ updates. 
The total time can thus be bounded by
\begin{align*}
&~\sum_{\ell=0}^{\log n} (\text{number of rank $2^\ell$ updates}) \cdot O(n^{\omega(1,\ell/\log n,1)}) \\
\le&~
\sum_{\ell= \lceil a\log n\rceil}^{\log n} O(T (C/\varepsilon_{mp}) 2^{-\ell/2} n^{\omega(1,\ell/\log n,1)} \log n ) \\
=&~
O(T (C/\varepsilon_{mp}) (n^{\omega-0.5} + n^{\omega(1,a,1)-a/2}) \log n )
\end{align*}
The last equality uses that $\omega(1,1,x)$ is a convex function, 
so the largest term of the sum must be the first or the last one.
If we assume $a \le \alpha$, then $n^{\omega(1,a,1)-a/2} = n^{2+o(1)-a/2}$, 
which leads to the complexity as stated in \Cref{lem:appendix:cohens_bound}.
\end{proof}
\fi

\begin{lemma}[{\cite{YeTM94,CohenLS19}}]\label{lem:feasible_LP}
Consider a linear program $\min_{Ax=b,x\geq0}c^{\top}x$ with $n$
variables and $d$ constraints. Assume that
\begin{enumerate}
\item Diameter of the polytope: For any $x\geq0$ with $Ax=b$, we have that $\|x\|_{1}\leq R$.
\item Lipschitz constant of the LP: $\|c\|_{\infty}\leq L$.
\end{enumerate}
For any $0<\gamma\leq1$, the modified linear program $\min_{\overline{A}\overline{x}=\overline{b},\overline{x}\geq0}\overline{c}^{\top}\overline{x}$
with
\ifdefined\SODAversion
$$
\overline{A}=\left[\begin{array}{ccc}
A & 0 & \frac{1}{R}b-A1_{n}\\
1_{n}^{\top} & 1 & 0\\
-1_{n}^{\top} & -1 & 0
\end{array}\right],\overline{b}=\left[\begin{array}{c}
\frac{1}{R}b\\
n+1\\
-(n+1)
\end{array}\right],$$
and $\overline{c}=\left[\begin{array}{c}
\gamma/L\cdot c\\
0\\
1
\end{array}\right]$
\else
\[
\overline{A}=\left[\begin{array}{ccc}
A & 0 & \frac{1}{R}b-A1_{n}\\
1_{n}^{\top} & 1 & 0\\
-1_{n}^{\top} & -1 & 0
\end{array}\right],\overline{b}=\left[\begin{array}{c}
\frac{1}{R}b\\
n+1\\
-(n+1)
\end{array}\right]\text{, and }\overline{c}=\left[\begin{array}{c}
\gamma/L\cdot c\\
0\\
1
\end{array}\right]
\]
\fi
satisfies the following:
\begin{enumerate}
\item $\overline{x}=\left[\begin{array}{c}
1_{n}\\
1\\
1
\end{array}\right]$, $\overline{y}=\left[\begin{array}{c}
0_{d}\\
0\\
1
\end{array}\right]$ and $\overline{s}=\left[\begin{array}{c}
1_{n}+\frac{\gamma}{L}\cdot c\\
1\\
1
\end{array}\right]$ are feasible primal dual vectors.
\item For any feasible primal dual vectors $(\overline{x},\overline{y},\overline{s})$
with $\sum_{i=1}^n \overline{x}_i\overline{s}_i \leq\gamma^{2}$, 
consider the vector $\hat{x}=R\cdot\overline{x}_{1:n}$
($\overline{x}_{1:n}$ is the first $n$ coordinates of $\overline{x}$)
is an approximate solution to the original linear program in the following
sense
\begin{align*}
c^{\top}\hat{x} & \leq\min_{Ax=b,x\geq0}c^{\top}x+LR\cdot\gamma,\\
\|A\hat{x}-b\|_{1} & \leq 2\gamma \cdot \left( R \sum_{i,j} |A_{i,j}| + \|b\|_{1} \right),\\
\hat{x} & \geq0.
\end{align*}
\end{enumerate}
\end{lemma}

\section{Projection Maintenance via Dynamic Linear System Solvers}
\label{app:matrix_inverse}

The data-structure from \cite{Sankowski04,BrandNS19} can maintain the solution to the following linear system:
Let $M$ be a non-singular $n \times n$ matrix and let $b$ be an $n$-dimensional vector.
Then the data-structures can maintain $M^{-1}b$ while supporting changing any entry of $M$ or $b$ in $O(n^{1.529})$ time.
This differs from the problem we must solve for the linear system \eqref{eq:interior_point_system}, where we must maintain $Pv$ for $P = \sqrt{X/S} A^\top (A\frac{X}{S}A)^{-1}A\sqrt{X/S}$ and the updates change entries of $X$ and $S$.
However, even though the structure seems very different, one can maintain $Pv$ via the following reduction:

\begin{lemma}\label{lem:blackbox_reduction}
Let $A$ be a $d \times n$ matrix of rank $d$ and let $U$ be an $n \times n$ diagonal matrix with non-zero diagonal entries.
Then
\begin{align*}
&~\left(\begin{array}{cccc}
U^{-1} & A^\top             & \sqrt{U}^{-1} & 0\\
A      & 0                  & 0             & 0\\
0      & 0                  & -I            & 0\\
(\sqrt{U}^{-1})^\top & 0    & 0             & -I
\end{array}\right)^{-1}
\left(\begin{array}{cccc}
0_n \\
0_n \\
v \\
1_n
\end{array}\right)\\
=&~
\left(\begin{array}{cccc}
*\\ 
*\\ 
*\\ 
\sqrt{U}A^\top(AUA^\top)^{-1}\sqrt{U}v
\end{array}\right),
\end{align*}
where $*$ represents some entries that do not care about.

\end{lemma}
\ifdefined\SODAversion
The proof for \Cref{lem:blackbox_reduction} can be found in the full version.
Via \Cref{lem:blackbox_reduction}, we can maintain the projection $Pv$ 
by using the data-structure of \cite{Sankowski04} 
that maintains $M^{-1}b$ by changing the diagonal entries of the $U^{-1}$ and $\sqrt{U}^{-1}$ blocks.
This way we are able to maintain $\sqrt{\tilde{U}} A^\top (A\tilde{U}A)^{-1}A\sqrt{\tilde{U}} f(\tilde{v})$ similar to \Cref{lem:projection_maintenance},
where $\tilde{U} = \diag(\tilde{u})$ and $\tilde{v}$ are approximate variants of the input parameters $u$ and $v$.
\else
We can thus maintain $Pv$ by using a data-structure that maintains $M^{-1}b$ by changing the diagonal entries of the $U^{-1}$ and $\sqrt{U}^{-1}$ blocks.

\begin{proof}[Proof of \Cref{lem:blackbox_reduction}]
The inverse of a $\text{two-blocks} \times \text{two-blocks}$ matrix is given by
\begin{align*}
\left(\begin{array}{cc}
Q & R\\
S & T
\end{array}\right)^{-1}
=
\left(\begin{array}{cc}
Q^{-1} + Q^{-1}R(T-SQ^{-1}R)^{-1}SQ^{-1} & -Q^{-1}R(T-SQ^{-1}R)^{-1} \\
-(T-SQ^{-1}R)^{-1}SQ^{-1} & (T-SQ^{-1}R)^{-1}
\end{array}\right)
\end{align*}
If $Q = U^{-1}$, $T = 0$, $R = A^\top$, $T = A$, then the matrix has full-rank (i.e. it is invertible) and the top-left block of the inverse is $U + UA^\top(AUA)^{-1}A^\top U$.
Further, consider the following block-matrix and its inverse:
\begin{align*}
\left(\begin{array}{ccc}
M & N & 0\\
0 & -I & 0\\
N^\top & 0 & -I
\end{array}\right)^{-1}
=
\left(\begin{array}{ccc}
M^{-1} & M^{-1}N & 0 \\
0 & -I & 0 \\
N^\top M^{-1} & N^\top M^{-1} N & -I
\end{array}\right)
\end{align*}
When $M$ is the previous block-matrix and $N$ is the $(n+d) \times n$ block-matrix $(\sqrt{U}^{-1}, 0_{n\times d})^\top$, then the matrix is exactly the one given in \Cref{lem:blackbox_reduction} and the bottom-center block of the inverse is
$$
N^\top M^{-1} N 
=
\sqrt{U}^{-1} (U + UA^\top(AUA)^{-1}A^\top U) \sqrt{U}^{-1}
=
I + \sqrt{U} A^\top(AUA)^{-1}A^\top \sqrt{U}.
$$
Let $C$ be this $(3n+d) \times (3n+d)$ block-matrix specified in \Cref{lem:blackbox_reduction} and let $b = (0_n,0_d,v,1_n)$ be an $(3n+d)$-dimensional vector, then the bottom $n$ coordinates of $C^{-1}b$ are exactly \linebreak $\sqrt{U} A^\top(AUA)^{-1}A^\top \sqrt{U} v$.
\end{proof}

One can use the data-structure of \cite{Sankowski04}
to maintain $\sqrt{\tilde{U}} A^\top (A\tilde{U}A)^{-1}A\sqrt{\tilde{U}} f(\tilde{v})$ similar to \Cref{lem:projection_maintenance},
where $\tilde{U} = \diag(\tilde{u})$ and $\tilde{v}$ are approximate variants of the input parameters $u$ and $v$.
\fi
Whenever some entry of $\tilde{U}$ or $\tilde{v}$ must be changed, because the approximation no longer holds,
the algorithm of \cite{Sankowski04} spends $O(n^{1.529})$ time per changed entry of $\tilde{U}$ and $\tilde{v}$.
This is not yet fast enough for our purposes, 
because when using this data-structure inside our linear program solver,
up to $\Omega(n)$ entries might be changed throughout the entire runtime of the solver.
Thus one would require $\Omega(n^{2.529})$ time for the solver.

By applying the complexity analyzsis of \cite{CohenLS19} to this data-structure, 
one can achieve the same amortized complexity as in \Cref{lem:projection_maintenance}.
We now briefly outline how this is done.

Per iteration of the linear system solver,
more than one entry of $\tilde{u}$ and $\tilde{v}$ may have to be changed.
This can be interpreted as a so called \emph{batch-update}, 
and the complexity for batch-updates was already analyzed in \cite{BrandNS19},
but again the focus was on worst-case complexity.
Both data-structure from \cite{Sankowski04} and \cite{BrandNS19} had the property,
that the data-structure would become slower the more updates they received.
This issue was fixed by re-initializing the data-structure in fixed intervals.
The core new idea of \cite{CohenLS19} is a new strategy for this re-initialization:
They wait until $n^a$ many entries of $\tilde{u}$ must be changed 
(see \cref{line:entry_condition} of \Cref{alg:projection_maintenance}),
and then they change preemptively a few more entries 
(see \cref{line:extra_entries}).

Applying the same reset strategy to \cite{Sankowski04,BrandNS19} then results in the same complexity as \Cref{lem:projection_maintenance}.
Indeed the resulting data-structure is essentially identical to \Cref{lem:projection_maintenance}/\Cref{alg:projection_maintenance},
because all these algorithms are just exploiting the Sherman-Morrison-Woodbury identity.

%\input{appendix_phi.tex}
%\newpage
%\input{appendix_original_algorithm.tex}
%\newpage

\bibliographystyle{alpha}
\bibliography{bibliography}

\end{document}